\documentclass{article}
\usepackage{graphicx}
\usepackage{color}
\usepackage{algorithm}
\usepackage{algpseudocode} 

\usepackage{amsmath}
\usepackage{amsthm}
\usepackage{amssymb}

\usepackage{graphicx}
\usepackage{color}
\usepackage[T1]{fontenc}            
\usepackage[utf8]{inputenc}

\usepackage[english]{babel}
\usepackage{mathptmx} 

\newtheorem{Thm}{Theorem}

\newtheorem{lem}[Thm]{Lemma}
\newtheorem{cor}[Thm]{Corollary}

\newcommand{\ceil}[1]{\left\lceil #1 \right\rceil}

\newcommand{\drop}[1]{} % Comment multi-line

\newcommand{\RM}{{\sc rm}}

\newcommand{\equals}{\stackrel{\operatorname{def}}{=}}
\newcommand{\sig}[2]{\sum_{#1}^{#2}}

\newcommand{\cefr}[2]{\left\lceil\frac{#1}{#2}\right\rceil}

\newcommand{\lef}{\left\{}% pre-defined
\newcommand{\ri}{\right\}}

\newcommand{\infer}{\Longrightarrow}

\newcommand{\beginProof}[1]{\noindent{\textbf Proof:}}

\newcommand{\finishProof}[1]%
{~\newline \noindent{\bf \qed\   End  proof} (of {\bf #1}).\newline}

\newtheorem{Def}{Definition}

\newtheorem{Th}{Theorem}
\newtheorem{Le}{Lemma}
\newtheorem{Coro}{Corollary}
\newtheorem{pro}{Proof:}
\newtheorem{Rmk}{Remark}
\def\sssec{\subsubsection}

\def\begdef{\begin{Def}  }
\def\enddef{\end{Def}}
\def\begth{\begin{Th} }
\def\endth{\end{Th}}
\def\begle{\begin{Le}  }
\def\endle{\end{Le}}
\def\begco{\begin{Coro}  }
\def\endco{\end{Coro}}
\def\begpro{\begin{proof} }
\def\endpro{\end{proof}}
\def\begex{\begin{example} }
\def\endex{\end{example}}
\def\begrmq{\begin{Rmk}  }
\def\endrmq{\end{Rmk}}
\def\begdes{\begin{description} \item[] }
\def\enddes{\end{description}}

\def\begeq{\begin{equation}  }
\def\endeq{\end{equation}}
\def\begeqno{\begin{equation*}  }
\def\endeqno{\end{equation*}}
\def\begarr{\begin{eqnarray}}
\def\endarr {\end{eqnarray}}
\def\begarrno{\begin{eqnarray*}}
\def\endarrno {\end{eqnarray*}}
\def\begsubeq{\begin{subequations}  }
\def\endsubeq{\end{subequations}}
\def\begsubeqno{\begin{subequations*}  }
\def\endsubeqno{\end{subequations*}}

\newcommand {\rle}[1]{Lemma \ref{#1}}

\newcommand {\req}[1]{Eq. (\ref{#1})}

\title{ Exact Polynomial Time Algorithm for the Response Time Analysis of Harmonic Tasks with Constrained Release Jitter}

\author{ 
Thi Huyen Chau Nguyen, \\ 
               Department of Information Technology, \\ Thang Long University (TLU), Hanoi, Vietnam\\
							 chaunth@thanglong.edu.vn
							\and
							Werner Grass,\\ 
							 Faculty of Computer Science and Mathematics,\\ University of Passau, Germany\\
              grass@fim.uni-passau.de
							\and 
							Klaus Jansen,\\ 
						  Department of Computer Science, \\Christian-Albrechts-University Kiel, Germany\\
							kj@informatik.uni-kiel.de
}

\begin{document} 
\maketitle

\begin{abstract} In some important application areas of hard real-time systems, preemptive sporadic tasks with harmonic periods and constraint deadlines running upon a uni-processor platform play an important role. We propose a new algorithm for determining the exact \emph{worst-case response time} for a task that has a lower computational complexity (linear in the number of tasks) than the known algorithm developed for the same system class. We also allow the task executions to start delayed due to release jitter if they are within certain value ranges. For checking if these constraints are met we define a constraint programming problem that has a special structure and can be solved with heuristic components in a time that is linear in the task number. If the check determines the admissibility of the jitter values, the linear time algorithm can be used to determine the \emph{worst-case response time} also for jitter-aware systems.
\end{abstract}
%\keywords{real-time systems \and response-time analysis \and workload function \and fixed point iteration \and release jitter \and integer programs with two variables per constraint }

% !TEX root = harmonic_27.11.2019.tex

\section{Introduction}
Hard real-time embedded systems must deliver  functional correct results  related to their initiating events within specified time limits. Such systems are usually modelled as a  composition of a finite number of recurrent tasks with the tasks releasing a potentially infinite sequence of jobs.  In the often used sporadic task model, the jobs arrive at a time distance that is greater than or equal to the inter-arrival time (called period), which thus represents an important task parameter. The processing of a job must be completed at the latest with the relative deadline of the associated task. An important step in the design of such a system is therefore the scheduling analysis, with which compliance with the time conditions is checked,
for the implementation of which further system properties must be introduced. 

In this paper we consider task executions by a single processor, a fixed priority task system and we allow a task being preempted in order to perform a higher priority task. Deadlines may be constrained by values lower as or equal to the corresponding period. A common method of scheduling analysis for these characteristics is response time analysis (RTA) \cite{Joseph1986},\cite{Audsley1993}. 

More recently, real-time systems with harmonic tasks ( the periods are pairs of integer multiples of each other) have received increased attention. This is due in part to the fact that harmonic task systems have at least two advantages over systems with arbitrary periods:   the processor utilization may be larger than in the general case and the \emph{worst-case response times} for the different tasks can be determined in polynomial time \cite{BoMa2013}  whereas in the general case RTA is pseudo-polynomial in the representation of the task system. In the literature we have several case studies in the most important applications fields like avionics \cite{Eisenbrand2010SolvingAA}, automotive \cite{anssi:cea-01779313}, industrial controllers \cite{d7c4f1a1dce54ef9a8883089d808cec5}, robotics \cite{DBLP:conf/rtas/ShihGGCS03}  where harmonic periods are used. If the given periods are not a priori harmonic, they can be made harmonic according to certain criteria from a set of non-harmonic periods with associated allowable value ranges, an appropriate objective function or to  satisfy  end-to-end latency requirements \cite{7176034},\cite{Davare:2007:POH:1278480.1278553},\cite{Mohaqeqi:2016:PFO:2997465.2997490}.

The release jitter of a task is the maximum difference between the arrival times and the release times  over all jobs of this task and may extend its \emph{worst-case response time} \cite{Audsley1993}. The combination of harmonic periods with rate monotonic prioritization leads to a reduction of the jitter problem, as both release jitter and execution time variation can be kept small since every job execution of a task is started at the same time distance from the lower period limit. With an arbitrary prioritization however this advantage no longer exists.

The release jitter concept can also be used to replicate other phenomena that have a corresponding effect on response times as the following two examples show. This makes release jitter all the more important in response time analysis. 

In \cite{TINDELL19951163}  response time analysis introduced for fixed priority scheduling on a uni-processor has been adapted and applied to the scheduling of messages on Controller Area Networks (CAN). 
%Each message is queued and is ready for transmission within a time interval $\left[0,J_m\right]$ after its triggering event, which occurs at a multiple of the message period. 
Instead of a release jitter, we now have a queue jitter of the message with the same effect on the delivery time of that message as the release jitter on the end time of a job.  
%In \cite{Davis2012ControllerAN} this method is improved by applying it to practically relevant queuing methods and arbitrary deadlines. 
 
A real-time job can suspend itself while waiting for an activity to complete. The dynamic self-suspension model allows a job of  a task to suspend itself at any time instance before it finishes as long as its worst-case self-suspension time is not exceeded. This property of real-time system may be modeled by a virtual jitter as discussed e.g. in \cite{DBLP:conf/rtcsa/ChenBH017}.

In the following we first introduce an iterative method to determine the exact \emph{worst-case response times} of harmonic tasks without considering release jitter. This method can be extended to the case that  all tasks have the same release jitter.  Finally, we show that this method can also be used for variable release jitter, provided that the jitter values fulfill certain restrictions, which we check with a linear-time algorithm.

\subsection{Related work}
In 1973, Liu and Layland \cite{Liu1973} had generalized the result on priority assignment of \cite{Fin67} to demonstrate the optimality of Rate Monotonic scheduling (\RM). They also presented a simple sufficient schedulability test for periodic fixed priorities tasks under \ \RM \ and the assumption that the deadline of a task is equal to its period. The sufficient test does not give an answer to the question whether  sets with $n$ tasks that lead to a higher total processor utilization ($>n( 2^{1/n}-1)$) can actually be scheduled or not. Kuo and Mok \cite{KuoMok1991} have shown that it is sufficient to create harmonic periods in order to eliminate uncertainty about schedulability. In this case it is sufficient to keep the total processor utilization $\leq 1$ in order to schedule the task system with the schedule policy considered in \cite{Liu1973}. 

In \cite{d7c4f1a1dce54ef9a8883089d808cec5} Xu et al. also consider harmonic task sets, but this time the \emph{worst-case response times} $R_i$ of the tasks are determined with a binary search process. Thus constraint deadlines $D_i$ with $R_i \leq D_i < T_i$  ($T_i$ denote the period of task $\tau_i$) can also be allowed. This method takes advantage of the fact that in the case of rate-monotonic prioritization of harmonic tasks, the start times of the jobs of a task $\tau_i$ always have the same distance to the previous period which is defined by the \emph{worst-case response time} of the task with the next higher priority. 

Bonifaci et al \cite{BoMa2013} no longer assume that the priorities decrease with longer periods (\RM), but allow any fixed priorities that are not dependent on any other task parameters. The basic task and scheduling model is the same as in our approach but the schedulability test is quite different. In order to determine the worst-case response time of a task $\tau_n$, they first arrange the tasks according to non-increasing periods $\tau_1\dots\tau_{n-1}$ such that after this reordering $\tau_1$ has the largest and $\tau_{n-1}$ the lowest period. They have shown that the response time for $\tau_n$ must be in the interval $[0,C_n T_1]$ where $C_n$ is the worst-case execution time of task $\tau_n$. This interval length is now gradually reduced to the task periods $T_2\dots T_{n-1}$, whereby they have to search for the right position of the smaller interval in the potentially larger  predecessor interval. In contrast, we follow the standard approach in which a modified fixed point iteration is carried out for determining the response time $R_n = C_n +\sum_{j<n}C_j\left\lceil R_n/T_j\right\rceil$, which can now be carried out in exactly $n$ iteration steps because of the harmonic periods.

The approaches described so far for the handling of task systems with harmonic periods do not allow a model extension to take release jitter into account. Rather, one must resort to methods that have been developed for the treatment of arbitrary periods. Audsley \cite{Audsley1993} and Tindell \cite{Tin94b} have definded the response time of task systems with jitter and Sjodin and Hansson \cite{Sjodin1998} have proved that  fixed point iteration can be applied for determining  the response time. We first use our approach to determine the response time when all tasks have the same maximum jitter. We then investigate task systems in which each task can have different jitter and specify restrictions for the jitters so that the fixed point determination can also be carried out with jitter-aware task systems in linear time. 

\subsection{This research}
The first aim of this research is to develop an algorithm that  determines the exact worst-case response time for fixed priority preemptive sporadic harmonic tasks with \emph{constrained deadlines} running on an uni-processor platform. Although these properties are equal with that in \cite{BoMa2013} our algorithm has  a lower computational complexity. It is based on the standard RTA approach which performs a fixed point iteration on the basis of the \emph{processor demand function} which takes into account the worst-case execution time of the examined task as well as its preemptions by higher-priority tasks (\emph{total interference}). In contrast to the standard approach we present a parametric  approximation of this total preemption time by higher priority tasks that contributes to the response time of the task considered. This approximation proceeds in $n$ phases of fine-tuning to get the exact \emph{total interference} hence arriving to the exact response time. 

The second objective is to include possible release jitter of the tasks. The necessary modification of the algorithm for jitter-free tasks is straight forward if we assume that all tasks have the same jitter. Then the approximations have only to be corrected by an additional jitter term. To handle the more general jitter-aware case  we introduce a different calculation rule for the preemption time of a task by higher prioritized tasks which, however, gives the same fixed point. This new formula has a certain formal similarity with the standard formula for jitter-aware systems. If the jitter of the task with the smallest period is the largest and the other jitters fulfill further constraints, we take the largest jitter as a constant jitter for all tasks and use the algorithm introduced for this case to determine the worst-case response time. Finally, we also allow other jitter values, but have to check whether certain constraints are met and have to determine the constant jitter value that is used to determine the worst-case response time. For checking purposes we define a constraint programming problem that has a special structure and can be solved with heuristic components in a time that is linear in $n$.

\subsection{Organization} We formally define the terminology, notation and task model in Section \ref{sec:model}.
In Section \ref{sec:prelresult}, we present our new algorithm for getting the \emph{worst-case response time} for a task in a time that is linear in $n$ assuming that the higher priority tasks are ordered by non-increasing periods. The correctness of the algorithm is proved in Section \ref{sec:proof}. In the rest of the paper we consider jitter-aware systems. In Subsection \ref{samereljit} we begin with modifying the algorithm introduced in Section \ref{sec:prelresult} for systems with the same jitter for all tasks. The new formula to determine the preemption time by higher priority tasks is introduced in Subsection \ref{sec:diff_models} and it is shown that the fixed-point iteration results in the same worst-case response time as the usually used formula. In Subsection \ref{sec:restricted}, we apply the result to task systems where the task with the lowest period has the largest jitter. Finally, we loose the restrictions on jitter and define a constraint programming problem in Subsection \ref{sec:looserestjit} and introduce an algorithm to solve it in Subsection \ref{constprog}.

\section{System model and background}\label{sec:def} \label{sec:model}

In this work, we analyze a set $\Gamma=\left\{\tau_1, \tau_2, \dots, \tau_n \right\}$ of
$n$ hard real-time sporadic \emph{tasks}, each one releasing a sequence of
\emph{jobs}.  Task $\tau_i$ is characterized by:
\begin{itemize}
\item a minimum interarrival time $T_i$ (that we call \emph{period},
  in short) between the arrival of two consecutive jobs,
\item a worst-case \emph{execution time} $C_i$, 
 and
\item a \emph{relative deadline} $D_i$.

\end{itemize}
In Section \ref{sec:SwJit} we will extend this model by release jitter $J_i$.
The task periods are assumed to be harmonic that is $T_i$ divides $T_j$ or vice versa $\left(T_i| T_j \right.$ or $\left. T_j|T_i\right)$. 
All task parameters are positive integer numbers. Notice that, by
properly multiplying all the parameters by an integer, rational
numbers are also allowed.  We assume \emph{constrained deadlines} i.e., $C_i \leq D_i\leq
T_i$. The ratio $U_i = C_i/T_i$ denotes the
\emph{utilization} of task $\tau_i$, that is, the fraction of time
required by $\tau_i$ to execute.

At the time instants denoted by $a_{i,j}$, the $i$-th task arrives and it is released for execution at a time $r_{i,j}\geq a_{i,j}$. A released task requests the execution of its $j$-th \emph{job} for an amount $C_i$ of time. 

The maximum difference $r_{i,j}- a_{i,j}$ over all $j$ is called release jitter $J_i$ and we start our presentation with the assumption of $J_i$ being 0 for all $1\leq i \leq n$.  

Two consecutive arrivals of the same task cannot be separated by less than $T_i$, that is,
\[
\forall i j,\qquad  a_{i,j+1}\geq a_{i,j}+T_i.
\]
We denote the \emph{finishing time} of the $j$-th job of the $i$-th
task by $f_{i,j}$. The \emph{worst-case response time} $R_i$ of a task is \cite{Davis2008}
\begin{equation}
  \label{eq:defRespTime}
  R_i = \max_{j=1,2,\ldots}\{f_{i, j} - r_{i, j}\} 
\end{equation}
A task set is said to be \emph{schedulable} when the maximum time period between the release and the finishing time of task $\tau_i$ is lower than the relative deadline  \cite{Davis2008}:
\[
\forall i,\qquad R_i\leq D_i.
\]

In this paper we assume that tasks are scheduled over a single
processor by preemptive Fixed Priorities (FP). The tasks are ordered
by decreasing priority: $\tau_i$ has higher priority than $\tau_j$ if
and only if $i<j$. Also, we use  abbreviated notations for the sum of  utilizations of tasks with successive indexes. We set $U_{\iota\ldots\kappa} = \sum_{i=\iota\ldots\kappa}U_i$. Correspondingly we denote $\sum_{i=\iota\ldots\kappa}C_i$ with $C_{\iota\ldots\kappa}$.

Also, we recall some basic notions related to fixed-priority
scheduling. In 1990, Lehoczky~\cite{Lehoczky1990} introduced the
notion of \emph{level-$i$ busy period}, which represents the intervals
of time when any among the higher priority tasks is running,
in the \emph{critical instant}.  

In jitter-free systems the \emph{critical instant} occurs when all tasks $\tau_i$ are simultaneously released ~\cite{Liu1973}. Without loss of generality such an instant is set equal to zero. In order to check the schedulability of a jitter-free task system, it is therefore sufficient to test the response times for the first jobs on compliance with the condition $R_i \leq D_i$.

For simplification in the notation, from now on  we consider the \emph{worst-case response time} of task $\tau_n$ but  our results could easily be applicable for all $i < n$. 

The \emph{total interference} $I_{n-1}\left(t\right)$ describes the amount of time that is taken for executing the  tasks with a higher priority then $\tau_n$ during the time interval $[0,t)$. 
\begin{equation}\label{eq:0}
I_{n-1}\left(t\right) = \sum_{i=1}^{n-1} C_i\cdot \left\lceil \frac{t}{T_i}\right\rceil
\end{equation} 
The time period $t-I_{n-1}\left(t\right)$ is therefore left during the interval $[0,t)$ for executing task $\tau_n$. The total time demand for a complete execution of the $n$-th task is given by the \emph{processor demand function}:
\begin{equation} \label{eq:W} 
W_n(t) = C_n + I_{n-1}(t) = C_n + \sig {i=1}{n-1} \cefr{t} {T_i}C_i
\end{equation}
The \emph{worst-case response time} is the point in time  at which $t-I_{n-1}\left(t\right) = C_n$. We therefore determine the \emph{worst-case response time} $R_n$ as the least fixed point \cite{Audsley1993}: 
\begin{equation} \label{eq:R}
R_n \equals \min \lef t| W_n(t) =t\ri
\end{equation}

\section{Preliminary results}\label{sec:prelresult}
 According to \eqref{eq:R}  and as proven in \cite{Sjodin1998} $R_n$ may be determined by an iterative technique starting with $R_{n}^{(0)}$ and  producing the values  $R_{n}^{(1)}, R_{n}^{(2)}, R_{n}^{(3)}, ...$ 
and approximating $R_n$ by applying the recurrence: 
\begin{equation}
R_{n}^{(0)}=\frac{C_{n}}{1-U_{1\ldots n-1}}\ \ \ \ \ 
R_{n}^{(k)} = C_{n}+ \sum_{i=1}^{n-1} C_i\cdot\left\lceil \frac{R_{n}^{(k-1)}}{T_i} \right\rceil
\end{equation}
 The iteration stops when $R_{n}^{(k)}=R_{n}^{(k-1)}$.
Although  the iteration converges for $U_{1\ldots n} < 1$ the number of iteration steps can be very high (pseudo-polynomial complexity). 

One of the main results of our paper is the introduction a completely different sequence of exactly $n$ approximations to the true value of $R_{n}$. It is presented in Theorem \ref{thm1}.
 
\textcolor[rgb]{0,0,0}{In preparation of the theorem, we introduce Lemmas that justify the admissibility of a task reordering. Such rearrangements were also made in \cite{BoMa2013}  and \cite{Bini2015} }.

% !TEX root = harmonic_27.11.2019.tex

For this purpose, we introduce the following lemma: 

\begle \label{le:I} The order in which the tasks $\tau_j$
with $j <n$  are executed is immaterial for the total interference of these higher priority tasks to $\tau_n$.
\endle
\begpro By construction of \req{eq:0}, we have in any time interval $[0,t]$, the total interference by  higher priority tasks to $\tau_n$ is $I_{n-1}(t) = \sum_{i=1}^{n-1} C_i\cdot \left\lceil \frac{t}{T_i}\right\rceil$. Since for any given $t$, the reorder of $n-1$ higher priority tasks of $\tau_n$ simply equates the reorder of $n-1$ terms of the accumulate sum, the interference time $I_{n-1} (t)$ remains the same. The lemma follows. 
\endpro

From this lemma, we obtain that the interference of higher priority tasks to $R_{n}^{(k)}$, i.e., $I_{n-1} (R_{n}^{(k-1)})$, is independent of the order in which tasks $\tau_i$ with $i <n$ are executed. This leads us to the idea of computing the \emph{worst-case response time} of $\tau_n$ by rearranging the order of its higher priority tasks. 

\begco \label{co:I} The order in which the tasks $\tau_j$
with $j <n$  are executed is immaterial for calculating the worst-case response time of task $\tau_n$.
\endco
\begpro Directly from \rle {le:I}, for any order in which the tasks $\tau_j$ with $j <n$  are executed, the interference time $I_{n-1} (t)$ remains the same with any given $t$. Consequently, the recursive equation $t= C_n + I_{n-1} (t)$ would obtain the same solution $t= R_n$. Hence to compute $R_n$ we could choose an arbitrary order of these higher priority tasks and the corollary is proved. 
\endpro

Also note that when this method is successively applied for all $\tau_i$ in the system, for each round of computation of $R_i$ since the priority of the $i$-th task as well as the set of its higher priorities remained unchanged, the corresponding re-ordering will be transparent to the $(i+1)$-th task, i.e., it is not a task priority re-assignment and only pre-process preparation for our \emph{worst-case response time} analysis. 

Now in order to keep the calculation of the indices simple in the various processing steps described below, we choose an inverse rate-monotonic order. To formally describe this reordering we introduce a bijective mapping 
\begin{equation}\label{eq:map}
\pi: 1\ldots n-1 \rightarrow 1\ldots n-1,
\end{equation} 
in which $\pi(i) = k$ signifies that task $\tau_k$ with priority $k$ is at position $i$ in the new order. The reverse rate monotonic order satisfies the condition that for all $i < j$ period $T_{\pi(j)}$ divides the period $T_{\pi(i)}$ having the priorities $\pi(i)$ and \textcolor[rgb]{0,0,0}{$\pi(j)$}, respectively.

\begin{Thm}\label{thm1} 
We are given a set of $n$ harmonic tasks in reverse rate monotonic order. Then the least fixed point of the equation
\begin{equation}\label{eq:Rnd}
R_n=C_{n}+\sum_{i=1}^{n-1}C_{\pi(i)}\left\lceil R_n/T_{\pi(i)}\right\rceil
\end{equation}
 can be obtained by applying the iterative formula:
\begin{equation}\label{eq:bc}
\widetilde{R}_n^{(0)}=\frac{C_{n}}{1-U_{\pi(1)\ldots \pi(n-1)}}	
\end{equation}
\begin{multline}\label{eq:itst}
	1\leq i \leq n-1,\ 	\  \widetilde{R}_n^{(i)}=\widetilde{R}_n^{(i-1)}+\frac{ C_{\pi(i)} \left(\left\lceil \frac{\widetilde{R}_n^{(i-1)}}{T_{\pi(i)}}\right\rceil-\frac{\widetilde{R}_n^{(i-1)}}{T_{\pi(i)}}\right)}{1-U_{\pi(i+1)\ldots \pi(n-1)}}
\end{multline}
we finally get $R_n = \widetilde{R}_n^{(n-1)}$.
\end{Thm}
The iteration can be stopped once $\widetilde{R}_n^{(i-1)}$ is an integer multiple of $T_{\pi(i)}$ i.e., if $\widetilde{R}_n^{(i)}= \widetilde{R}_n^{(i-1)}$ holds. For then $\widetilde{R}_n^{(i-1)}$ is also an integer multiple of $T_{\pi(j)} \leq T_{\pi(i)}$.

Eq. \eqref{eq:Rnd} describes the usual form of the recursion to determine $R_n$. Note that after changing the order of the tasks the value of the sum remains the same, i.e., we could also write the terms in the sum as $C_{i}\left\lceil t/T_{i}\right\rceil$ without changing the result of the sum.

Before we prove  Theorem 1 in the next section, let us introduce some properties of the result.
The calculation of $R_n$ ends at the latest after $n$  steps and therefore has a linear complexity. Note that only $n-1$ ceiling functions have to  be applied. In comparison, the search algorithm in \cite{BoMa2013} has the complexity $\mathcal O(n\cdot log (T_{\pi (1)}))$ with $T_{\pi(1)} = max_{1\leq i \leq n}(T_{\pi(i)})$. In \cite{d7c4f1a1dce54ef9a8883089d808cec5} an algorithm has been proposed that is also based on a binary search but with a reduced complexity  $\mathcal O(log(T_{\pi(1)})-log(T_{\pi(n)}))$ to compute the response time of  task $\tau_n$ if the priorities are rate monotonic, i.e. decrease with increasing period length. If $T_{\pi(n)}=2^{n{-}1}$ our algorithm and that in \cite{d7c4f1a1dce54ef9a8883089d808cec5} have about the same complexity but our algorithm can be applied to arbitrary fixed-priorities. 
Considering \cite{BoMa2013} and our algorithm, we must add the time required for sorting to get the complexity of the complete algorithm. 

%\begin{exmp}
%\begin{equation*}
%\widetilde{R}_4^{(0)}=\frac{C_{\pi(4)}}{1-U_{\pi(1)..\pi(3)}}	
%\end{equation*}
%\begin{equation*}
% \widetilde{R}_4^{(1)}=\frac{ \widetilde{R}_4^{(0)}\left(1-U_{\pi(1)..\pi(3)}\right)+C_{\pi(1)} \left\lceil \frac{\widetilde{R}_4^{(0)}}{T_{\pi(1)}}\right\rceil}{1-U_{\pi(2)..\pi(3)}}
%\end{equation*}
%\begin{equation*}
 %\widetilde{R}_4^{(2)}=\frac{ \widetilde{R}_4^{(1)}\left(1-U_{\pi(2)..\pi(3)}\right)+C_{\pi(2)} \left\lceil \frac{\widetilde{R}_4^{(1)}}{T_{\pi(2)}}\right\rceil}{1-U_{\pi(3)}}
%\end{equation*}
%\begin{equation*}
%	 \widetilde{R}_4^{(3)}= \widetilde{R}_4^{(2)}\left(1-U_{\pi(3)}\right)+C_{\pi(3)} \left\lceil \frac{\widetilde{R}_4^{(2)}}{T_{\pi(3)}}\right\rceil
%\end{equation*}
%\end{exmp}

\section{Proof of Theorem \ref{thm1}}\label{sec:proof}
In order to obtain the result of Theorem \ref{thm1}, in this section we would present a parametric  approximation of the \emph{total interference} of higher priority tasks that contributes to the response time of $\tau_n$. This approximation proceeds in $n$ phases of fine-tuning to get the exact \emph{total interference} hence arriving to the exact \emph{worst-case response time}. The main part of this section is the prove that this fine-tuning can be performed in an inductive fashion, of which each phase now has a constant computational complexity.   

 For this purpose, first of all we introduce a set of functions  $\widetilde{I_i}\left(t\right)$ with changing $i$ represents varying degrees of approximation of the \emph{total interference} $I_{n-1}\left(t\right)$. To define such functions, we partition the task set $\Gamma-\tau_{n}$ into two disjoint subsets $\Gamma^{\left\lceil \right\rceil}_{i} $ and $\Gamma^{/}_{i}$. We start with $\Gamma^{\left\lceil \right\rceil}_{0}=\emptyset $ and $\Gamma^{/}_{0}=\Gamma-\tau_{n}$  and terminate with $\Gamma^{\left\lceil \right\rceil}_{n-1}=\Gamma-\tau_{n} $ and $\Gamma^{/}_{n-1}=\emptyset$.
 For the functions formed therebetween we produce $\Gamma^{\left\lceil \right\rceil}_{i}=\left\{\tau_{\pi(1)}\dots.\tau_{\pi(i)}\right\}$ and $\Gamma^{/}_{i}=\left\{\tau_{\pi(i+1)}\ldots \tau_{\pi(n-1)}\right\}$  with $0$ $\leq i\leq n-1$. The indexes of the tasks in the set $\Gamma^{/}_{i}$ determine which addends in the definition equation  of $I_{n-1}\left(t\right)$ are replaced by their linear lower bounds  i.e., the rule $x\leq\left\lceil x\right\rceil$ is applied. Its approximations are defined as follows:
\begin{equation}\label{eq:I}
0\leq i \leq n-1, \ \ \ \widetilde{I}_{i}\left(t\right) = \sum_{i+1\leq j \leq n-1} U_{\pi(j)}\cdot t + \sum_{1 \leq j \leq i} C_{\pi(j)}\left\lceil \frac{t}{T_{\pi(j)}}\right\rceil
\end{equation}
 The left sum is formed by the elements of $\Gamma^{/}_{i}=\left\{\tau_{\pi(i+1)}\ldots \tau_{\pi(n-1)}\right\}$ and the right sum by the elements of $\Gamma^{\left\lceil \right\rceil}_{i}=\left\{\tau_{\pi(1)}\dots.\tau_{\pi(i)}\right\}$. Also note that by this construction, $I_{n-1}\left(t\right)=\widetilde{I}_{n-1}\left(t\right)$.

The difference of two functions with immediately successive indexes is
 \[\widetilde{I_i}\left(t\right)-\widetilde{I}_{i-1}\left(t\right)= C_{\pi(i)}\left\lceil \frac{t}{T_{\pi(i)}}\right\rceil-U_{\pi(i)}\cdot t \]

It follows that the two functions are equal for all times that are multiples of $T_{\pi(i)}$ and that the maximum distance over all time instances  is less than $ C_{\pi(i)}$. 
Fig. \ref{fig:1} shows an example of two functions $C_{n}+\widetilde{I_{i}}\left(t\right)$ and $C_{n}+\widetilde{I_i}_{-1}\left(t\right)$. We are interested in the points of intersection with the identify function $Id(t)=t$ and want to construct the solution of $t_i= C_{n}+\widetilde{I_i}\left(t_i\right)$ knowing the solution of $t_{i-1}= C_{n}+\widetilde{I_i}_{-1}\left(t_{i-1}\right)$. We observe in Fig. \ref{fig:1}, that $\widetilde{I_i}\left(t\right)$ is linear in the time interval $\nu\cdot T_{\pi(i)} < t \leq (\nu+1)\cdot T_{\pi(i)}$ with $\nu \in \mathbb{N}$ and the two points of intersection (at the start point of the vertical arrow and the end point of the horizontal arrow) are within the same period of task $\tau_{\pi(i)}$.
\begin{figure*}[!htp]
\centering
\includegraphics[width=0.60\textwidth]{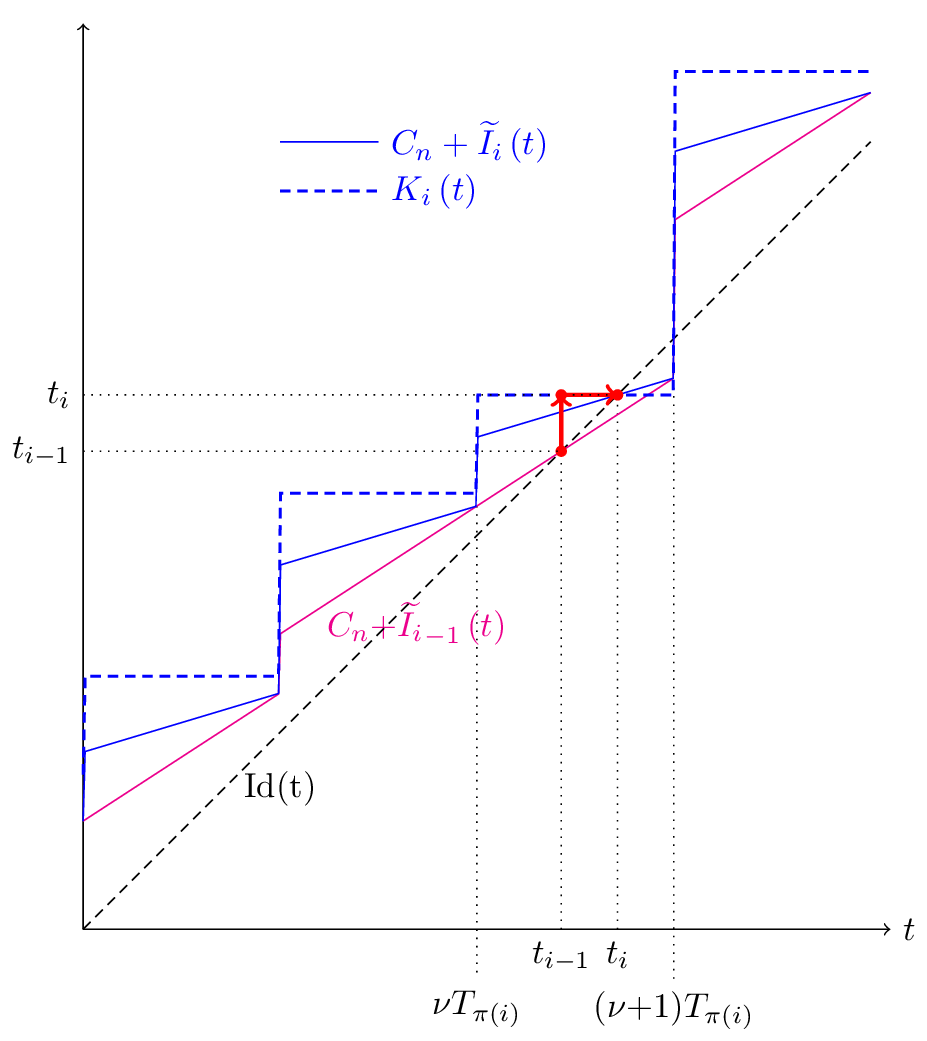}
\caption{The figure shows an example of functions $C_n+\widetilde{I_i}\left(t\right)$, $C_n+\widetilde{I_i}_{-1}\left(t\right)$,  $K_i\left(t\right)$, and $Id\left(t\right)$ as well as the solutions of $t_{i-1}=C_{n}+\widetilde{I_{i}}_{-1}\left(t\right)$ and $t_i=K_i(t_i)$}
	\label{fig:1}	
\end{figure*}

To simplify the process of finding $t_i$ when knowing $t_{i-1}$ we introduce another set of functions which have no subterms that are linear in time and would significantly lengthen the number of iterations until the fixed point is found. We obtain the more suitable equality by manipulating the set of equations $t=C_{n}+\widetilde{I_i}\left(t\right)$.  In doing so, the equations are solved for the time variable $t$ as far as possible i.e., leaving the ceiling terms. We get $t=K_i\left(t\right)$ with
\begin{equation}\label{eq:t_i}
0\leq i \leq n-1,\  K_i\left(t\right) =  \frac{C_{n}+\sum_{1 \leq j \leq i} C_{\pi(j)}\left\lceil \frac{t}{T_{\pi(j)}}\right\rceil}{1- U_{\pi(i+1)\ldots \pi(n-1)}} 
\end{equation}
The solution defined by any of the equations $t_i=K_i\left(t_i\right)$ is equal to that of the corresponding equation $t_i=C_{n}+\widetilde{I_i}\left(t_i\right)$ as shown in the following Lemma.
\begin{lem}\label{fixed_point_ident}
$t_i=C_{n}+\widetilde{I_i}\left(t_i\right) \Leftrightarrow t_i=K_i(t_i)$
\end{lem}
\begin{proof}
Let be 
\begin{equation}
t_i= C_n+ \widetilde{I}_{i}\left(t_i\right) = C_n+\sum_{i+1\leq j \leq n-1} U_{\pi(j)}\cdot t_i + \sum_{1 \leq j \leq i} C_{\pi(j)}\left\lceil \frac{t_i}{T_{\pi(j)}}\right\rceil  \tag{*}
\end{equation}
Using some algebra we get
\begin{equation}
t_i=  \frac{C_n +\sum_{1 \leq j \leq i} C_{\pi(j)}\left\lceil \frac{t_i}{T_{\pi(j)}}\right\rceil}{1-\sum_{i+1\leq j \leq n-1} U_{\pi(j)} }=K_i(t_i) \tag{**}
\end{equation}
which proofs the $\Rightarrow$ direction. We can also start with eq. (**) and make the reverse conversion to eq. (*). This proves the $\Leftarrow$ direction.
%\qed
\end{proof}

In Fig.\ref{fig:1} $K_i\left(t\right)$  (dashed line) and $C_{n}+\widetilde{I_i}\left(t\right)$  (solid line) have the same point of intersection with the identity function.  Note also that in the figure $K_i(t)$ is constant in the interval $\nu\cdot T_{\pi(i)} < t \leq (\nu+1)\cdot T_{\pi(i)}$ hence $K_i(t_{i-1})=K_i(t_i)$. 

The use of the functions \eqref{eq:t_i} is not new. In \cite{Lu2006} these functions are used to reduce the number of iterations applying the RTA method.

The solution of any equation $t_i =K_i\left(t_i\right)$ can be found by an iteration:
\begin{equation}\label{eq:fp}
  t_i^{\left(k{+}1\right)} =  \frac{C_{n}+\sum_{1 \leq j \leq i} C_{\pi(j)}\left\lceil \frac{t_i^{(k)}}{T_{\pi(j)}}\right\rceil}{1- U_{\pi(i+1)\ldots \pi(n-1)}}, \ \ \ 0\leq i \leq n-1\  
\end{equation}
ending when we obtained the fixed point $t_i^{\left(k+1\right)}=t_i^{\left(k\right)}$, which is a solution of $t_i=K_i(t_i)$.

For the iterative calculation of the fixed point of $t_i=K_i(t_i)$, it is important that this fixed point is greater than or equal to the fixed point of the equation $t_{i-1}=K_{i-1}(t_{i-1})$, as shown in the following Lemma.
\begin{lem}\label{fixed_point_comparison}
 $t_i \geq t_{i-1}$.
\end{lem}
\begin{proof}
$t_{i-1}$ is the least fixed point of $t=K_{i-1}(t)$. For all $t< t_{i-1}$ we therefore have $t<K_{i-1}(t)$. We write $K_i(t)$ in terms of $K_{i-1}(t)$:  
\begin{equation*}
K_i(t) = \frac{K_{i-1}(t)(1-U_{\pi(i)\dots\pi(n-1)}) + C_{\pi(i)}\left\lceil t/T_{\pi(i)}\right\rceil }{1-U_{\pi(i+1)\dots\pi(n-1)}}
\end{equation*}
For $t<t_{i-1}$ we have:
\begin{equation*}
K_i(t) > \frac{t(1-U_{\pi(i)\dots\pi(n-1)}) + C_{\pi(i)}\left\lceil t/T_{\pi(i)}\right\rceil }{1-U_{\pi(i+1)\dots\pi(n-1)}}
\end{equation*}
With $U_{\pi(i)\dots\pi(n-1)} =U_{\pi(i)}+ U_{\pi(i+1)\dots\pi(n-1)}$ it is
\begin{equation*}
K_i(t) >t+ \frac{-t U_{\pi(i)} + C_{\pi(i)}\left\lceil t/T_{\pi(i)}\right\rceil }{1-U_{\pi(i+1)\dots\pi(n-1)}}\geq t
\end{equation*}
$t=K_i(t)$ therefore has no fixed point in the interval $[0,t_{i-1})$, so that $t_i \geq t_{i-1}$ applies.
%\qed
\end{proof}

This Lemma ensures that we can start the search for the fixed point $t_i$ with the fixed point $t_{i-1}\leq t_i$.

Theorem \ref{thm1} states that $t_i=K_i(t_{i-1})$ i.e., if we set $t_i^{\left(0\right)} = t_{i-1}$, only one iteration step is required and  $t_i^{\left(1\right)}$ is the solution we are searching for.
The proof of Theorem \ref{thm1} is done by induction and structured in the following way:
\begin{enumerate}
	\item Introducing the base case in Lemma \ref{eq:BC} 
	\item In the inductive hypothesis we assume that we have a solution of $t_{i-1} = K_{i-1}\left(t_{i-1}\right)$. We start the iterative procedure for getting a solution of $t_i = K_i\left(t_i\right)$ with the initial value $t^{\left(0\right)}_i=t_{i-1}$ and use the hypothesis in order to get a simplified version of the equation for $t_{i}$ (Lemma \ref{lem:iter2}) 
	\item Finally we will give a solution for $t_i=K_i\left(t_i\right)$ and prove its validity (Lemma \ref{lem:indstep}).
\end{enumerate}

We start with the base case.
\begin{lem}\label{eq:BC}
The base case of the inductive proof is defined by $i=0$. 
\begin{equation}\label{eq:bc2}
t_{0}=\frac{C_{n}}{1-U_{\pi(1)\ldots \pi(n-1)}}
\end{equation}
\end{lem}
\begin{proof}
This follows immediately from $t_{0}=K_{0}\left(t_{0}\right)$ (see Eq. \eqref{eq:t_i} with $i=0$).
%\qed
 \end{proof}
Supposing we already had a solution  of the  Eq. \eqref{eq:t_i} for the index $i-1$, i.e. a solution of
\begin{equation*}
 t_{i-1} =  \frac{C_{n}+\sum_{1 \leq j \leq i-1} C_{\pi(j)}\left\lceil \frac{t_{i-1}}{T_{\pi(j)}}\right\rceil}{1- U_{\pi(i)\ldots \pi(n-1)}} 
\end{equation*}
for the variable $t_{i-1}$.

Now, we calculate $t_i$ as a function of $t_{i-1}$ on the assumption that we start the iteration for getting $t_i$ with the initial value $t_i^{\left(0\right)}=t_{i-1}$.
\begin{lem}\label{lem:iter2}
We assume that we have found a solution for the  equation $t_{i-1}=K_{i-1}\left(t_{i-1}\right)$. If we start the iteration for the determination of $t_i$ with this value $t_{i-1}$, we can use the equation
\begin{equation}\label{eq:iter2}
t_i^{(1)} = \frac{t_{i-1}\cdot\left(1- U_{\pi(i)\ldots \pi(n-1)}\right) + C_{\pi(i)}\cdot \left\lceil \frac{t_{i-1}}{T_{\pi(i)}} \right\rceil}{1- U_{\pi(i+1)\ldots \pi(n-1)}} 
\end{equation}
\end{lem} 
\begin{proof}
By the induction hypothesis we have
\begin{equation*}
  t_{i-1} =  \frac{C_{n}+\sum_{1 \leq j \leq i-1} C_{\pi(j)}\left\lceil \frac{t_{i-1}}{T_{\pi(j)}}\right\rceil}{1- U_{\pi(i)\ldots \pi(n-1)}} \tag{*}
\end{equation*}
After multiplying both sides of the equation (*) by  $1- U_{\pi(i)\ldots \pi(n-1)}$ we get:
\begin{equation}\label{eq:t_i+1}
C_{n}+\sum_{1 \leq j \leq i-1} C_{\pi(j)}\left\lceil\frac{t_{i-1}}{T_{\pi(j)}}\right\rceil=t_{i-1}\cdot\left(1- U_{\pi(i)\ldots \pi(n-1)}\right)  
\end{equation}
This equation allows us to substitute a subterm of a term in the form of the LHS by the RHS.

We now consider the iteration for the determination of $t_i$ starting with $t_{i-1}$. In Eq. \eqref{eq:fp} we split the sum of the ceiling terms for $j=1$ to $j=i$ in the numerator into the sum from $j=1$ to $j=i-1$ and a separately written addend with $j=i$:
\begin{equation*}
t_i^{\left(1\right)} =  \frac{C_{n}+\sum_{1 \leq j \leq i-1} C_{\pi(j)}\left\lceil \frac{t_{i-1}}{T_{\pi(j)}}\right\rceil+C_{\pi(i)}\left\lceil \frac{t_{i-1}}{T_{\pi(i)}}\right\rceil}{1- U_{\pi(i+1)\dots\pi(n-1)}} 
\end{equation*}
 Observing \eqref{eq:t_i+1} to simplify the numerator we get the Lemma.
%\qed
\end{proof}

Note that Eq. \eqref{eq:iter2} can also be written as
\begin{equation}\label{eq:iter3}
t_i^{(1)} = t_{i-1} +\frac{-t_{i-1}\cdot U_{\pi(i)} + C_{\pi(i)}\cdot \left\lceil \frac{t_{i-1}}{T_{\pi(i)}} \right\rceil}{1- U_{\pi(i+1)\ldots \pi(n-1)}} = t_{i-1}  + \frac{C_{\pi(i)}\cdot \left(\left\lceil \frac{t_{i-1}}{T_{\pi(i)}} \right\rceil-\frac{t_{i-1}}{T_{\pi(i)}}\right)}{1- U_{\pi(i+1)\ldots \pi(n-1)}} 
\end{equation}
so that $t_i^{(1)} \geq t_{i-1}$.

After these preparations we can now take the decisive third step in the proof of Theorem \ref{thm1}.
First, a basic property of nested ceiling functions is recalled from Lemma~\ref{lem:nestedCeil} in \cite{GrNg}.

\begin{lem}
  \label{lem:nestedCeil}
  Let be $x,z$ positive rational numbers with $0<x\leq z$. Then,
  \begin{equation}
    \ceil{x+\left(1-\frac{x}{z}\right)\ceil{z}} = \ceil{z}
    \label{eq:ceilingProp}
  \end{equation}
\end{lem}

This Lemma allows to substitute the complex LHS by the simple RHS in a corresponding term.

To complete the proof of Theorem \ref{thm1} we have to show that the value of $t_i^{(1)}$  is a solution of the equation $t_i=K_i(t_i)$. With reference to Fig. (1), we can elucidate the underlying idea as a preparation of the Lemma \ref{lem:indstep}. The start point of the vertical  arrow has the coordinates $\left(t_{i-1}, t_{i-1}\right)$ and the end point of the horizontal arrow has the coordinates $\left(t_{i}, t_{i}\right)$. The function $K_i(t)$ is constant for all $t\in  \left[t_{i-1},t_i\right]$ hence $K_i(t_{i-1})=K_i(t_{i})$.

\begin{lem}\label{lem:indstep}
We assume the validity of \eqref{eq:iter2} (induction hypothesis) and get
\begin{equation*}
 t_i^{\left(1\right)}=K_i\left(t_{i-1}\right) \Rightarrow t_i^{\left(1\right)} = K_i\left(t_{i}^{\left(1\right)}\right)\Rightarrow t_i=t_i^{\left(1\right)}
\end{equation*}
\end{lem}
\begin{proof}
The proof is done in three steps: 
\begin{enumerate}
	\item $K_i\left(t\right)$ is constant in the left open time interval 

$$I=		\left]T_{\pi\left(i\right)}\left(\left\lceil \frac{t_{i-1}}{T_{\pi\left(i\right)}}\right\rceil{-1}\right),\  T_{\pi\left(i\right)}\left\lceil \frac{ t_{i-1}}{T_{\pi\left(i\right)}}\right\rceil \right]$$
	
\begin{proof}
Eq. \eqref{eq:t_i} contains only ceiling terms with a denominator $T_{\pi\left(j\right)}$ that is a multiple of $T_{\pi\left(i\right)}$ and the numerator $t$. All these ceiling terms are therefore constant in the mentioned interval. There are no other time dependent terms.

\end{proof}
	
	\item The two intances of time $t_{i-1} \in I$ and $t_i^{\left(1\right)} \in I$ with $t_i^{\left(1\right)} \geq t_{i-1}$ lie in the same period of the task $\tau_{\pi\left(i\right)}$, i.e. $T_{\pi\left(i\right)}\left\lceil \frac{ t_{i-1}}{T_{\pi\left(i\right)}}\right\rceil$ = $T_{\pi\left(i\right)}\left\lceil \frac{ t_{i}^{\left(1\right)}}{T_{\pi\left(i\right)}}\right\rceil$.
This is true if the ceiling terms are equal.
\begin{proof}
In Lemma \ref{lem:iter2} we have shown how we can use a simple term to determine $t_i^{(1)}$. From this we get:
\begin{equation*}
\left\lceil \frac{t_i^{\left(1\right)}}{T_{\pi\left(i\right)}}\right\rceil = \left\lceil  \frac{t_{i-1}\cdot\left(1- U_{\pi(i)\ldots \pi(n-1)}\right) + C_{\pi(i)}\cdot \left\lceil \frac{t_{i-1}}{T_{\pi(i})} \right\rceil}{T_{\pi\left(i\right)}\left(1- U_{\pi(i+1)\ldots \pi(n-1)}\right)} \right\rceil
\end{equation*}
By setting  
\begin{equation} 
x =  \frac{t_{i-1}\cdot\left(1- U_{\pi(i)\ldots \pi(n-1)}\right) }{T_{\pi\left(i\right)}\left(1- U_{\pi(i+1)\ldots \pi(n-1)}\right)} \tag{*}
\end{equation} and
\begin{equation} z =\frac{t_{i-1}}{T_{\pi(i)}}\tag{**},
\end{equation} 
 from the property of the nested ceiling of Lemma \ref{lem:nestedCeil}, it follows that
\begin{equation}
\left\lceil \frac{t_i^{\left(1\right)}}{T_{\pi\left(i\right)}}\right\rceil =\left\lceil z\right\rceil= \left\lceil \frac{t_{i-1}}{T_{\pi(i)}} \right\rceil \tag{***}
\end{equation}
The Lemma is applicable in this case, since  
\begin{itemize}
	\item with $t_0=C_n/(1-U_{\pi(1)\dots\pi(n-1)})>0$ and Lemma 3 we have $0< t_0 < t_1 < \dots < t_{i-1}$ and $z>0$.
	\item with $1-U_{\pi(i)\ldots \pi(n-1)} >0$ and $t_{i-1} > 0$ we have $x>0$
  \item 
for the factor of the inner ceiling applies: 
\begin{equation} 
C_{\pi\left(i\right)}/\left(T_{\pi\left(i\right)}\left(1-U_{\pi\left(i+1\right)\ldots  \pi\left(n-1\right)}\right)\right) = 1-x/z \tag{****}
\end{equation} 
since by (*) and (**)
\begin{equation*}
1-x/z = 1-\frac{t_{i-1}\cdot\left(1- U_{\pi(i)\ldots \pi(n-1)}\right) }{T_{\pi\left(i\right)}\left(1- U_{\pi(i+1)\ldots \pi(n-1)}\right)}\frac{T_{\pi(i)}}{t_{i-1}} 
\end{equation*}
The RHS can be transformed into:
\begin{equation*}
1-\frac{1- U_{\pi(i)\ldots \pi(n-1)} }{1- U_{\pi(i+1)\ldots \pi(n-1)}}=\frac{U_{\pi\left(i\right)}}{1- U_{\pi(i+1)\ldots \pi(n-1)}}>0
\end{equation*}
The RHS is equivalent to the LHS of (****) and is $> 0$. 
\end{itemize}

Since $K_i\left(t\right)$ is constant $\forall t \in I$ and we have for the time instances $t_{i-1} \in I$ and $t_i^{(1)} \in I $ we get  $t_i^{(1)}=K_i\left(t_{i-1}\right)=K_i(t_i^{(1)})$.

\end{proof}
	\item Finally, $t_i=t_i^{(1)}$
\begin{proof}
The point in time as defined in Eq. \eqref{eq:iter2} $t_i^{(1)}=K_i(t_i^{(1)})$   is a fixed point of the equation $t=K_i(t)$. Furthermore we started the fixed point iteration with the time instant $t_{i-1}$ which is lower than $t_i^{(1)}$. Therefore $t_i^{(1)}$ is the least fixed point i.e. $t_i=t_i^{(1)}$. 
%By Lemma \ref{fixed_point_comparison} $t_i \geq t_{i-1} \Rightarrow t_i=t_i^{(1)}$. 
By Lemma \ref{fixed_point_ident} it is also the least fixed point of $t = C_n + \widetilde{I}_i(t)$.

\end{proof}
\end{enumerate}
%\qed
\end{proof}
We are now able to give a proof of Theorem \ref{thm1}:
\begin{proof}  Theorem \ref{thm1}:
The equations \eqref{eq:bc} and \eqref{eq:itst} are equivalent to \eqref{eq:bc2} and \eqref{eq:iter3}, respectively if we substitute $t_0$ by $\widetilde{R}_n^{(0)}$ and $t_i$ by $\widetilde{R}_n^{(k)}$.
By Lemma \ref{lem:indstep} we know that $\widetilde{R}_n^{(k)}$ is a solution of \eqref{eq:itst}. Also by construction of \eqref{eq:I} $\widetilde{I}_{n-1}(t)= {I}_{n-1}(t)$ hence $\widetilde{R}_n^{(n-1)}$ is the exact \emph{worst-case response time} $R_n$. 
%\qed
\end{proof}

\section{Systems with Jitter}\label{sec:SwJit}
When  release jitters have to be considered, the harmonic tasks lose the important property that all discontinuities of the \emph{total interference function} are restricted to the times which are multiples of the smallest period $T_{\pi\left(n-1\right)}$. This may cause that no polynomial algorithm exists for the determination of the \emph{worst-case response time}, and one has to resort to the algorithm for general tasks. In this section, we show that we can handle jitter that meets certain limits.

Release jitter models the delay between the arrival time of a job and the time the job is released. To determine the \emph{worst-case response time}, we can assume that the jobs of a task follow each other with the minimum distance. The arrival times are therefore multiples of the smallest inter-arrival time \cite{Baruah90preemptivelyscheduling}.  That is, the $i$-th job of a task $\tau_{j}$ arriving at the time $a_{j,i}=i T_j$ is released within the time interval $i T_j \leq r_{j,i} \leq i T_j + J_j$ where $J_j$ denotes the maximum jitter of the task over all jobs. As consequence, the time interval between consecutive releases of a task may be lower than $T_j$ and 
 the critical situation for a task $\tau_j$ arises when it is released together with all higher-priority tasks. Different points in time can therefore be \emph{critical instances} for the individual tasks.

Audsley  et al. \cite{Audsley1993} and Tindell et al. \cite{Tin94b} have discussed the release jitter problem for real time tasks in detail and proposed a modification of the \emph{total interference function}
\begin{equation}\label{eq:jit}
J^{J_i}_{n-1}(t)=\sum_{i=1}^{n-1}C_i\left\lceil\frac{t+J_i}{T_i} \right\rceil
\end{equation}

which leads to the \emph{processor demand function}
\begin{equation*} 
W^{J_i}_{n}(t)=C_n + J^{J_i}_{n-1}(t) = C_n + \sum_{i=1}^{n-1}C_i\left\lceil\frac{t+J_i}{T_i} \right\rceil
\end{equation*}
(The index $J_i$ denotes functions, whose values depend on a variable jitter).
The \emph{worst-case response time} of task $\tau_n$ results as the smallest fixed point  \cite{Sjodin1998}
\begeq \label{eq:fixp_J_i}
R_n^{J_i} \equals \min \lef t| W_n^J{i}(t) =t\ri
\endeq 
For this task to be schedulable, it must fulfill the condition \cite{Davis2008}:
\begin{equation*}
R_n^{J_i} \leq D_n-J_n
\end{equation*}
Again, we can reorder the higher-priority tasks without changing the value of these functions, where the ordering rule for $\pi$ is extended by ordering tasks with equal periods according to increasing jitter. Ties are broken arbitrarily.
We develop our method for determining the corresponding \emph{worst-case response time} in four steps:
\begin{enumerate}
	\item Assuming the same jitter for all tasks.
	\item Shifting the times of discontinuities in a limited way without changing the time of the fixed point.
	\item Application of the former results to handle strongly restricted different jitters for different tasks (subsection \ref{sec:restricted})
	\item Loosening the restrictions introduced in subsection \ref{sec:restricted}.
\end{enumerate}
\subsection{Assuming the same release jitter for all tasks}\label{samereljit}
To prepare our actual result, we first assume that all tasks have the same jitter $J$. The discontinuities of the \emph{processor demand function} are now shifted from integer multiples of the smallest period in the jitter-free case to a time that exactly $J$ time units lie before these periods. 
\begin{equation}\label{eq:J1}
W_{n}^J\left(t\right)=C_n+\sum_{i=1}^{n-1} C_i\left\lceil \frac{t+J}{T_i}\right\rceil
\end{equation}
and the \emph{worst-case response time} is defined as
\begeq \label{eq:fixp_J}
R_n^J \equals \min \lef t| W_n^J(t) =t\ri
\endeq

The condition for a task set to be schedulable is now:
\begin{equation*}
R^J_n \leq D_n-J
\end{equation*}

 Because we consider in this section different \emph{total interference functions} which result in different \emph{worst-case response times} we introduce indexes to make the difference visible. This recurrence can be solved analogously as the recurrence without jitter. We follow therefore the line of reasoning in sections \ref{sec:prelresult} and \ref{sec:proof} and begin with defining the function $K^J_i(t)$ where we use a reordering according to the mapping function $\pi$ and for which we search the least fixed points.
\begin{equation}\label{eq:K^J}
 0\leq i \leq n-1,\  K^J_i\left(t\right) =  \frac{C_{n}+\sum_{1 \leq j \leq i} C_{\pi(j)}\left\lceil \frac{t+J}{T_{\pi(j)}}\right\rceil}{1- U_{\pi(i+1)\ldots \pi(n-1)}}
\end{equation}
\begin{lem}\label{thm2}
We are given a set of $n$ harmonic tasks  where the $n-1$ tasks of higher priority are  ordered according to the mapping funtion $\pi$. Then the least fixed point of the equation
\eqref{eq:fixp_J} can be obtained by iteration:
\begin{equation}\label{eq:bcwt}
\widetilde{R}_n^{J(0)}=\frac{C_{n}+J}{1-U_{\pi(1)\ldots \pi(n-1)}}	-J
\end{equation}
\begin{equation}\label{eq:itstwtJ}
	1\leq i \leq n-1,	\ \ \ \widetilde{R}_n^{J(i)}= \widetilde{R}_n^{J(i-1)}+\frac{-U_{\pi\left(i\right)}(\widetilde{R}_n^{J(i-1)}+J)+C_{\pi(i)} \left\lceil \frac{\widetilde{R}_n^{J(i-1)}+J}{T_{\pi(i)}}\right\rceil}{1-U_{\pi(i+1)\ldots \pi(n-1)}}
\end{equation}
we finally get $R_n^J = \widetilde{R}_n^{J(n-1)}$.
\end{lem}
\begin{proof}
The value $\widetilde{R}_n^{J(0)}$ is the solution of the equation $$\widetilde{R}_n^{J(0)}= C_n+U_{\pi(1)\dots\pi(n-1)}\left(\widetilde{R}_n^{J(0)}+J\right)$$ which is derived from \eqref{eq:J1} by replacing all ceiling functions by their arguments and which has been defined in \cite{Sjodin1998} as a valid starting point for a fixed point iteration. 

Note that with $J>0$ and $0<U_{\pi(1)\dots \pi(n-1)}<1$ we have $\widetilde{R}_n^{J(0)}>0$.

The central part of the proof corresponds to the steps of the proof of Lemma \ref{lem:indstep}.
\begin{enumerate}

	\item 
	The introduced jitter changes the time intervals $I$ with constant values of $K_i^J(t)$ to:
	\begin{equation*}
	I=		\left]T_{\pi\left(i\right)}\left(\left\lceil \frac{\widetilde{R}_n^{J(i-1)}+J}{T_{\pi\left(i\right)}}\right\rceil{-1}\right),\  T_{\pi\left(i\right)}\left\lceil \frac{ \widetilde{R}_n^{J(i-1)}+J}{T_{\pi\left(i\right)}}\right\rceil \right]
\end{equation*}
This means  the discontinuities are now shifted to the time instances $mT_{\pi(i)}-J=m'T_{\pi(n-1)}-J$ for $K_i((t)$ with $m, m' \in \mathbb{N}$ but the interval length has not changed.
\item Now we have to show that  $\left\lceil (\widetilde{R}_n^{J(i-1)}+J)/T_{\pi(i)}\right\rceil = \left\lceil (\widetilde{R}_n^{J(i)}+J)/T_{\pi(i)}\right\rceil$ holds what means that the two time instances $\widetilde{R}_n^{J(i-1)}$ and $\widetilde{R}_n^{J(i)}$ are in the same time interval with a constant value of $K_i^J(t)$ i.e. $K_i^J(\widetilde{R}_n^{J(i)}) = K_i^J(\widetilde{R}_n^{J(i-1)}) $.

We substitute the RHS of \eqref{eq:itstwtJ} for $\widetilde{R}_n^{J(i)}$ in $\left\lceil (\widetilde{R}_n^{J(i)}+J)/T_{\pi(i)}\right\rceil$ and show that this term can be simplified to $\left\lceil (\widetilde{R}_n^{J(i-1)}+J)/T_{\pi(i)}\right\rceil$. Again we use Lemma \ref{lem:nestedCeil}.

We set  
	
\begin{multline}
x=\frac{\widetilde{R}_n^{J(i)}+J}{T_{\pi(i)}} \overset{by \ \eqref{eq:itstwtJ}}{=} \frac{ \widetilde{R}_n^{J(i-1)}\left(1-U_{\pi(i)\ldots \pi(n-1)}\right)-U_{\pi\left(i\right)}J}{T_{\pi(i)}\left(1-U_{\pi(i+1)\ldots \pi(n-1)}\right)}+\frac{J}{T_{\pi(i)}}=\\ \frac{ \left(\widetilde{R}_n^{J(i-1)}+J\right)\left(1-U_{\pi(i)\ldots \pi(n-1)}\right)}{T_{\pi(i)}\left(1-U_{\pi(i+1)\ldots \pi(n-1)}\right)}\tag{*}
\end{multline}
and
\begin{equation}
z=\frac{\widetilde{R}_{n}^{J(i-1)}+J}{T_{\pi(i)}} \tag{**}
\end{equation}
With $\widetilde{R}_n^{J(i-1)}>\widetilde{R}_n^{J(0)}$ and $J>0$ we also have $x> 0$ and $z > 0$. From the property of the nested ceiling of Lemma \ref{lem:nestedCeil}, it follows that
\begin{equation}
\left\lceil \frac{\widetilde{R}_{n}^{J(i)}+J}{T_{\pi(i)}}\right\rceil=\left\lceil z\right\rceil=\left\lceil \frac{\widetilde{R}_{n}^{J(i-1)}+J}{T_{\pi(i)}}\right\rceil \tag{***}
\end{equation}
The Lemma is applicable in this case, since
 for the factor of the inner ceiling applies: 
\begin{equation} 
C_{\pi\left(i\right)}/\left(T_{\pi\left(i\right)}\left(1-U_{\pi\left(i+1\right)\ldots  \pi\left(n-1\right)}\right)\right) = 1-x/z \tag{****}
\end{equation} 
since by (*) and (**)
\begin{equation*}
x/z = \frac{(\widetilde{R}_n^{J(i-1)}+J )\left(1- U_{\pi(i)\ldots \pi(n-1)}\right) }{T_{\pi\left(i\right)}\left(1- U_{\pi(i+1)\ldots \pi(n-1)}\right)}\cdot\frac{T_{\pi(i)}}{\widetilde{R}_n^{J(i-1)}+J} 
\end{equation*}

 After simplification we also get:
\begin{equation*}
1-x/z =1- \frac{1- U_{\pi(i)\ldots \pi(n-1)}} {1- U_{\pi(i+1)\ldots \pi(n-1)}}
\end{equation*}
Using the equivalence $\left(1- U_{\pi(i+1)\ldots \pi(n-1)}\right)- U_{\pi(i)}\equiv \left(1- U_{\pi(i)\ldots \pi(n-1)}\right)$
and the definition $U_{\pi(i)}= C_{\pi(i)}/T_{\pi(i)}$ we have :
\begin{equation*}
1-x/z =\frac{ C_{\pi(i)}} {T_{\pi(i)}(1- U_{\pi(i+1)\ldots \pi(n-1)})}
\end{equation*}

\end{enumerate}
This proves that the  Lemma \ref{lem:nestedCeil} is applicable and Eq. (****) applies.
%\qed
\end{proof}
We can now give a first application of this method and start with introducing some definitions. We assume that the different tasks have different jitters and we are interested in the minimum and maximum values:  $$  J_{max} = \max\{ J_{\pi(i)} : 1 \le i \le n-1\}$$ and
 $$ J_{min} = \min\{ J_{\pi(i)} : 1 \le i \le n-1\} $$ 

We first set $J = J_{min}$, and then determine the {worst-case response time} which we denote by $R^{min}_n$  applying \eqref{eq:bcwt} and \eqref{eq:itstwtJ}. On the other hand, we set $J = J_{max}$ and obtain the \emph{worst-case response time} $R^{max}_n$.
\begin{lem}
We assume that we have different jitters for the different tasks and therefore have to consider the \emph{processor demand function}:
\begin{equation}\label{eq:equalJ}
W_{n-1}^{J_i}(t)=C_n+\sum_{i=1}^{n-1} C_{\pi(i)}\left\lceil (t+J_{\pi(i)})/T_{\pi(i)}\right\rceil
\end{equation}
The solution of the recursive equation 
\begin{equation}\label{eq:RJi}
R_n^{J_i}=W_{n}^{J_i}(R_n^{J_i})
\end{equation}
leads to the \emph{worst-case response time} for which we have
\begin{equation*}
R^{min}_n \leq R_n^{J_i} \leq R^{max}_n
\end{equation*}
 \end{lem}
\begin{proof} 
We compare the total sums in the equations for  determining  $R^{min}_n ,R_n^{J_i}$, and $R^{max}_n$ summand-wise. The constant term $C_n$ is equal in the three sums. For the summands with index i we have: $C_{\pi(i)}\left\lceil \frac{t+J_{min}}{T_{\pi(i)}}\right\rceil$,  $C_{\pi(i)}\left\lceil \frac{t+J_{\pi(i)}}{T_{\pi(i)}}\right\rceil$,  and \\ $C_{\pi(i)}\left\lceil \frac{t+J_{max}}{T_{\pi(i)}}\right\rceil$ with $0 \leq J_{min} \leq$  $J_{\pi(i)} \leq J_{max}$. With $x \leq y \Rightarrow \left\lceil x \right\rceil \leq \left\lceil y\right\rceil$ we get:
$\left\lceil \frac{t+J_{min}}{T_{\pi(i)}}\right\rceil \leq \left\lceil \frac{t+J_{\pi(i)}}{T_{\pi(i)}}\right\rceil \leq \left\lceil \frac{t+J_{max}}{T_{\pi(i)}}\right\rceil$. This is true for any $i$ and therefore it is $\forall t, \ \ W^{J_{min}}_{n}(t) \leq W^{J_i}_n(t) \leq W^{J_{max}}_n(t)$. The least intersections  of these functions with the identity function $Id(t)=t$ form the fixed points searched for.  This implies that the fixed points are ordered as the lemma states.
%\qed
\end{proof}
Unfortunately, the \emph{worst-case response times} for $J_{min}$ and $J_{max}$ may differ greatly. In either case, however, we can determine an upper limit and lower limit on $R_n^{J_i}$ in linear time.  
\subsection{ Different models for the \emph{total interference function}}\label{sec:diff_models}
\sssec {Preliminary results of the exclusion intervals} \label{ssec:exclusionIntervals}
The equation
\begin{equation}\label{eq:Inf}
I_{n-1}(R_n)=\sum_{i=1}^{n-1}C_{\pi(i)}\left\lceil R_n/T_{\pi(i)}\right\rceil
\end{equation} 
 does not describe the only model for the \emph{total interference} for $J=0$, especially in case of task sets with harmonic periods we can modify it. Our goal is to change the time instances of discontinuities individually for different tasks similar to the shift by jitter. We will show that the \emph{worst-case response time} using this different interference model  remains unchanged by this modification such that it can be computed by our method introduced above.  To derive this model and its implication we define time intervals which can not contain $R_n$ as an element and which we call the \emph{exclusion time intervals}. 

 In \eqref{eq:Inf}, the entire execution time of jobs of the task $\tau_{\pi(i)}$ is added
 immediately after the time instants which are multiples of $T_{\pi(i)}$ but
 in fact, the processor processes the activated jobs continuously.   That means that the time intervals  $m_iT_\pi{(i)} < t\leq m_iT_{\pi(i)}+C_{\pi(i)}$ can only be used by the  task $\tau_{\pi(i)}$ or a higher priority task but never to execute parts of the low priority task $\tau_n$.

Hence, the \emph{worst-case response time} for the task $\tau_n$ can not be in the left open time intervals $\left(m_iT_i,m_iT_i+C_i\right], \forall i\leq n-1, \forall m_i \in \mathbb{N} $. Rather it lies before or after such an interval.

When we consider tasks with harmonic periods, we can increase the \emph{exclusion time intervals}. At the time $m_iT_{\pi(i)}$ at least jobs of the tasks whose period divides $T_{\pi(i)}$  are released and must be executed according to their priority before the next time portion can be assigned to  the task $\tau_{n}$. Therefore $R_n$ can not be an element of the time intervals $\left(m_iT_{\pi(i)},m_iT_{\pi(i)}+C_{\pi(i)\ldots \pi(n-1)}\right],  1 \leq i\leq n-1, \forall m_i \in \mathbb{N}$. For the task with the lowest period $T_{\pi(n-1)}$ the \emph{exclusion time intervals} are $\left.\left(m_{n-1}T_{\pi(n-1)},m_{n-1}T_{\pi(n-1)}+C_{\pi(n-1)}\right.\right]$ $\forall m_{n-1} \in \mathbb{N}_0$.

Note that due to the commutativity of the addition $C_{\pi(i+1)\ldots\pi(n-1)}$ the order of the executed tasks after a period is not relevant.

In Figure \ref{fig:2} we give an example of a task system with $n=5$.  The upper solid line shows $C_5+I_{4}(R_5)$  whereas the lower solid line represents $C_5+\widehat{I}_{4}(R_5)$. $R_5$ can only be an element of the time intervals in which the dotted line meets the two solid lines. At each multiple of $T_{\pi(4)}$ an \emph{exclusion time interval}  begins which can not contain $R_5$ and which has different length depending on the tasks that are released at that time. At time $0$,  for example, all 5 tasks  are released resulting in an \emph{exclusion time interval} of length $C_{\pi(1)\dots.\pi(4)}$. The length of the  \emph{exclusion time interval} are shown for the multiples of $T_{\pi(4)}$. For the sake of simplicity, we have presented the task executions after the times $m_4 T_{\pi (4)}$ in the order of increasing periods. When the order of execution is changed, nothing changes in the course of the dotted line.

\begin{figure*}[!htp]
\centering
\includegraphics[width=0.90\textwidth]{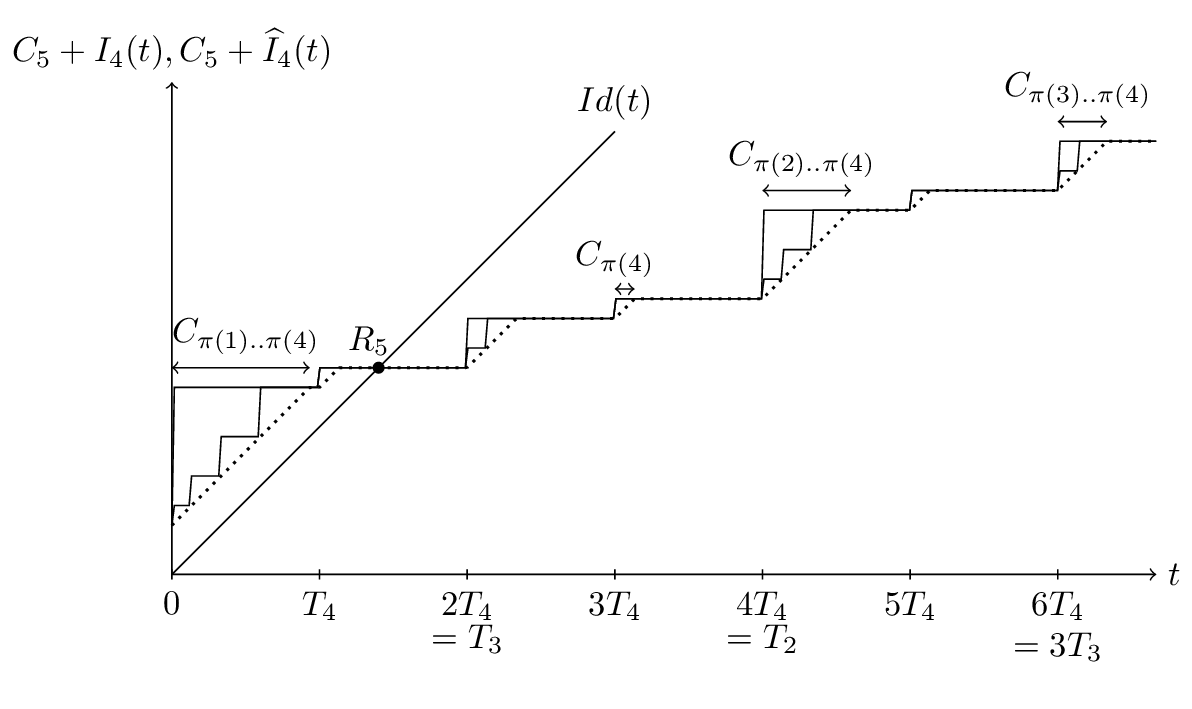}
\caption{The figure shows a 5 task system and explains that shifting activation times of tasks has no influence on the response time}
	\label{fig:2}	
\end{figure*}
We now can formally demonstrate this property of exclusion intervals for $R_n$ as follows: 
\begin{lem}\label{le:exclusionIntervalR}
\begin{equation*}
\forall i = 1\dots n-1, \forall m_i \in \mathbb{N},\ \ \ R_n \notin (m_iT_{\pi(i)}, m_iT_{\pi(i)} + C_{\pi(i)\dots \pi(n-1)}] 
\end{equation*} 
\end{lem}
\begpro Supposing that $\exists i \in [1,n-1], m_i \in \mathbb{N}$ such that $ m_iT_{\pi(i)} < R_n \leq m_iT_{\pi(i)} + C_{{\pi(i)}\dots {\pi(n-1)}} $, we will prove this Lemma by counterposition.

By definition of \req {eq:R}, since $W_n(t)$ is a non-decreasing function with $W_n(0)>0$, and $R_n$ is the first instant that $W_n(t)=t$, we should have for all $t<R_n$, $W_n(t)>t$, hence $W_n(m_iT_{\pi(i)})>m_iT_{\pi(i)}$. \hfill (*)\\
Also, we have: 

\begarr
W_n(R_n)-W_n(m_iT_{\pi(i)}) &=& \sig {j=1}{n-1} \left( \cefr{R_n}{T_{\pi(j)}} - \cefr {m_iT_{\pi(i)}}{T_{\pi(j)}} \right)C_{\pi(j)} \nonumber\\
&=&\sig {j=1}{i-1} \left( \cefr{R_n}{T_{\pi(j)}} - \cefr {m_iT_{\pi(i)}}{T_{\pi(j)}} \right)C_{\pi(j)} \nonumber\\
 && + \sig {j=i}{n-1} \left( \cefr{R_n}{T_{\pi(j)}} - \cefr {m_iT_{\pi(i)}}{T_{\pi(j)}} \right)C_{\pi(j)} \nonumber\\
&\geq& \sig {j=i}{n-1} \left( \cefr{R_n}{T_{\pi(j)}} - \cefr {m_iT_{\pi(i)}}{T_{\pi(j)}} \right)C_{\pi(j)} \label{eq:mT-R}
\endarr

By nature of our reordered harmonic systems,  for all $j=i\dots n-1$, we have $T_{\pi(j)}|T_{\pi(i)}$ hence $T_{\pi(j)}|m_iT_{\pi(i)}$ and $\cefr {m_iT_{\pi(i)}}{T_{\pi(j)}}=\frac {m_iT_{\pi(i)}}{T_{\pi(j)}} <\frac {R_n}{T_{\pi(j)}}\leq\cefr {R_n}{T_{\pi(j)}}$. Therefore $\forall j \in [i,n-1], \cefr {R_n}{T_{\pi(j)}}-\cefr {m_iT_{\pi(i)}}{T_{\pi(j)}} \geq 1$. Replacing these into \req{eq:mT-R} we have: 

\begarrno
W_n(R_n)-W_n(m_iT_{\pi(i)}) &\geq& \sig {j=i}{n-1} C_{\pi(j)} \\
W_n(R_n) &\geq& W_n(m_iT_{\pi(i)}) + \sig {j=i}{n-1} C_{\pi(j)} \\
R_n=W_n(R_n) &>& m_iT_{\pi(i)} + C_{{\pi(i)}\dots {\pi(n-1)}} \text { (by (*))}
\endarrno
which contradicts that $R_n \in (m_iT_{\pi(i)}, m_iT_{\pi(i)} + C_{{\pi(i)}\dots {\pi(n-1)}}] $. The Lemma is proved. 
\endpro

%This result will now be applied to tasks with different jitters.
\sssec{A novel model to compute the task worst-case response times} \label{ssec:anotherModel}
Our new model for calculating harmonic task system \emph{worst-case response time} allows for the treatment of release jitter by applying the algorithm in Lemma \ref{thm2}. In this subsection we start with proving the equivalence of the novel model with the standard model concerning the result of the corresponding fixed point iterations. 

With such a model $\widehat{I}_{n-1}(t)$ of the \emph{total interference}, we can now define another \emph{processor demand function} $\widehat{W}_{n}(t)$ as: 
\begeq \label{Wh}
\widehat{W}_{n}(t)  \equals C_n+ \widehat{I}_{n-1}(t) = C_n + \sig {i=1}{n-1} C_{\pi(i)}\cefr{t-C_{\pi(i+1)\dots \pi(n-1)}} {T_{\pi(i)}}
\endeq 
where $\forall i, C_{\pi(i+1)\dots \pi(n-1)} < T_{\pi(i)}$ because of the assumptions $U_{\pi(1)\dots \pi(n)}< 1$ and $T_{\pi(1)}\geq T_{\pi(2)} ... \geq T_{\pi(n-1)}$. We then obtain its least fixed point $\widehat{R}_{n}$ as: 
\begeq\label{Rh}
\widehat{R}_{n} \equals \min \lef t| \widehat{W}_{n}(t) =t\ri
\endeq
Note that the arguments of the ceiling function in \eqref{Wh} are similar to that of \eqref{eq:equalJ} with the exception that in \eqref{eq:equalJ} positive jitter values are added to the time variable $t$ whereas in \eqref{Wh} the sum of execution times is subtracted from $t$.

We will now demonstrate that this novel analysis model of  $\widehat{I}_{n-1}(t)$,  $\widehat{W}_{n}(t)$ and  $\widehat{R}_{n}$ is equivalent to the response time analysis of ${I}_{n-1}(t)$,  ${W}_{n}(t)$ and  ${R}_{n}$. For this purpose, we first present the following property of  $\widehat{R}_{n}$:

\begin{lem}\label{le:exclusionIntervalRh}
\begeq
\forall i = 1\dots n-1, \forall m_i \in \mathbb{N},\ \ \ \widehat{R}_{n} \notin (m_iT_{\pi(i)}, m_iT_{\pi(i)} + C_{{\pi(i)}\dots {\pi(n-1)}}] 
\endeq 
\end{lem}
\begpro
Supposing that $\exists i \in [1,n-1], m_i \in \mathbb{N}$ such that $ m_iT_{\pi(i)} < \widehat{R}_n \leq m_iT_{\pi(i)} + C_{{\pi(i)}\dots {\pi(n-1)}} $, we will prove this Lemma by counterposition.

Note also that since $\widehat{W}_{n}(t)$ is a non-decreasing function with $\widehat{W}_{n}(0)>0$, and $\widehat{R}_n$ is the first instant that $\widehat{W}_{n}(t)=t$, we should have for all $t<\widehat{R}_n, \widehat{W}_{n}(t)>t$. \hfill (*)\\ 

Now we will prove that $\widehat{R}_n> m_iT_{\pi(i)}+ C_{\pi(i)\dots\pi(n-1)}$ by induction. \\
1. First we prove the base case that $ \widehat{R}_n> m_iT_{\pi(i)}+C_{\pi(n-1)}$. 

Since by definition, $\widehat{W}_n(\widehat{R}_n) = \widehat{R}_n$, we have: 
\begin{multline}\label{eq:basech}
\widehat{R}_n-\widehat{W}_n(m_iT_{\pi(i)}) =  \widehat{W}_n(\widehat{R}_n)-\widehat{W}_n(m_iT_{\pi(i)}) \\
= \sig {j=1}{n-1} C_{\pi(j)}\left( \cefr{\widehat{R}_n-C_{\pi(j+1)\dots\pi(n-1)}}{T_{\pi(j)}} - \cefr {m_iT_{\pi(i)}- C_{\pi(j+1)\dots\pi(n-1)}}{T_{\pi(j)}} \right) 
\end{multline}
Let us consider the addends of the accumulated sum above: 

- For $j =1\dots n-2$: Since $\widehat{R}_n > m_iT_{\pi(i)}$ we have: 

\begeq \label{eq:jBeforeLast}
\drop{i=1\dotsn-2,}\cefr{\widehat{R}_n-C_{\pi(j+1)\dots\pi(n-1)}}{T_{\pi(j)}} - \cefr {m_iT_{\pi(i)}- C_{\pi(j+1)\dots\pi(n-1)}}{T_{\pi(j)}} \geq 0
 \endeq

- For $j=n-1$:  By nature of our reordered harmonic systems, with $i\leq n-1$, we have $T_{\pi(n-1)}|T_{\pi(i)}$ hence $T_{\pi(n-1)}|m_iT_{\pi(i)}$ and $\cefr {m_iT_{\pi(i)}}{T_{\pi(n-1)}}=\frac {m_iT_{\pi(i)}}{T_{\pi(n-1)}} <\frac {\widehat{R}_n}{T_{\pi(n-1)}}\leq\cefr {\widehat{R}_n}{T_{\pi(n-1)}}$. Therefore, with $j=n-1$:

 \begeq\label{eq:jLast}
\drop{j= n-1,} \cefr {\widehat{R}_n}{T_{\pi(j)}}-\cefr {m_iT_{\pi(i)}}{T_{{\pi(j)}}} \geq 1
\endeq

Applying \req{eq:jBeforeLast} and \req{eq:jLast} into Eq. \eqref{eq:basech}, we have $\widehat{R}_n - \widehat{W}_n(m_iT_{\pi(i)})\geq C_{\pi(n-1)}$. By (*), we also have $\widehat{W}_n(m_iT_{\pi(i)})>m_iT_{\pi(i)}$, hence $\widehat{R}_n \geq m_iT_{\pi(i)}+ C_{\pi(n-1)}$ and the base case is proved. \\
2. Now supposing we already had $\widehat{R}_n> m_iT_{\pi(i)}+ C_{\pi(k)\dots\pi(n-1)}$ with $i+1\leq k\leq n-1$ we will prove that $\widehat{R}_n>  m_iT_{\pi(i)}+ C_{\pi(k-1) \dots\pi(n-1)}$. We have:  

\begin{multline}
\widehat{R}_n-\widehat{W}_n(m_iT_{\pi(i)}+C_{\pi(k)\dots\pi(n-1)})=\widehat{W}_n(\widehat{R}_n)-\widehat{W}_n(m_iT_{\pi(i)}+C_{\pi(k)\dots\pi(n-1)}) \\
= \sig {j=1}{n-1}C_{\pi(j)}\left( \cefr{\widehat{R}_n-C_{\pi(j+1)\dots\pi(n-1)}}{T_{\pi(j)}} -\right. \\ \left.\cefr {m_iT_{\pi(i)} +C_{\pi(k)\dots\pi(n-1)}- C_{\pi(j+1)\dots\pi(n-1)}}{T_{\pi(j)}} \right) \label{eq:inductive}
\end{multline}

Let us consider the addends of the accumulated sum above: 

- For $j =1\dots n-1, j \neq k-1$: Since by the inductive hypothesis $\widehat{R}_n >  m_iT_{\pi(i)}+ C_{\pi(k)\dots\pi(n-1)}$, we have: 

\begeq \label{eq:jBeforeLastInductive}
 \drop{j =1\dots n-1, j \neq k-1,} \cefr{\widehat{R}_n-C_{\pi(j+1)\dots\pi(n-1)}}{T_{\pi(j)}} - \cefr {m_iT_{\pi(i)}+ C_{\pi(k)\dots\pi(n-1)}- C_{\pi(j+1)\dots\pi(n-1)}}{T_{\pi(j)}} \geq 0
 \endeq

- For $j=k-1$: By the inductive hypothesis, we have $\widehat{R}_n-C_{\pi(j+1)\dots\pi(n-1)}>m_iT_{\pi(i)}+ C_{\pi(k)\dots\pi(n-1)}- C_{\pi(j+1)\dots\pi(n-1)}= m_iT_{\pi(i)}$. Also, $i+1\leq k \infer k-1 \geq i$, then by nature of our reordered harmonic systems, we have $T_{\pi(k-1)}|T_{\pi(i)}$ hence $T_{\pi(k-1)}|m_iT_{\pi(i)}$ and $\cefr {m_iT_{\pi(i)}}{T_{\pi(k-1)}} =\frac {m_iT_{\pi(i)}}{T_{\pi(k-1)}}$. Combine all these we have $\frac{\widehat{R}_n-C_{\pi(j+1)\dots\pi(n-1)}}{T_{\pi(j)}} >\frac {m_iT_{\pi(i)}}{T_{\pi(k-1)}}=\cefr {m_iT_{\pi(i)}}{T_{\pi(k-1)}}= \cefr {m_iT_{\pi(i)}+ C_{\pi(k)\dots\pi(n-1)}- C_{\pi(j+1)\dots\pi(n-1)}}{T_{\pi(j)}}$. Finally, applying $\ceil{x}\geq x$, for $j=k-1$ we obtain:  

 \begeq\label{eq:jLastInductive}
\drop{j=k-1,} \cefr{\widehat{R}_n-C_{\pi(j+1)\dots\pi(n-1)}}{T_{\pi(j)}} - \cefr {m_iT_{\pi(i)}+ C_{\pi(k)\dots\pi(n-1)}- C_{\pi(j+1)\dots\pi(n-1)}}{T_{\pi(j)}} \geq 1
\endeq

 Applying \req{eq:jBeforeLastInductive} and \req{eq:jLastInductive} into \req{eq:inductive}, we have $\widehat{R}_n -\widehat{W}_n(m_iT_{\pi(i)}+ C_{\pi(k)\dots\pi(n-1)}) \geq  C_{\pi(k-1)}$. Also, by (*), $\widehat{W}_n(m_iT_{\pi(i)}+ C_{\pi(k)\dots\pi(n-1)})>m_iT_{\pi(i)}+ C_{\pi(k)\dots\pi(n-1)}$. Consequently, we have $\widehat{R}_n > m_iT_{\pi(i)}+ C_{\pi(k-1)\dots\pi(n-1)}$ and the inductive case is proved. 

Finally, applying the results above inductively from the base case until $k=i+1$, we have $\widehat{R}_n > m_iT_{\pi(i)}+ C_{\pi(i)\dots\pi(n-1)}$, which contradicts the counterpositive hypothesis and the Lemma is proved.  
\endpro

Now having this property of exclusion intervals for $\widehat{R}_{n}$, we could prove the correctness of the novel response time analysis: 

\begin{lem}\label{le:RRh}
\begeq
R_n = \widehat{R}_n
\endeq
\end{lem}
\begpro  We compare the associated ceiling functions composing $W_n(t)$ and $\widehat{W}_n(t)$:

 $$\left\lceil t/T_{\pi(j)}\right\rceil \\\ \text{and} \\\left\lceil (t-C_{\pi(j+1)\ldots \pi(n-1)})/T_{\pi(j)}\right\rceil$$

The two terms have the same value for $m_jT_{\pi(j)}+C_{\pi(j+1)\ldots \pi(n-1)} {<} t {\leq}  (m_j+1)T_{\pi(j)}$ for any $j \in[1, n-1]$.  \hfill (*)

By \rle{le:exclusionIntervalR}, the times of (*) are contained in the permissible value range of ${R}_n$. Therefore for $t = {R}_n$ we obtain $ W_n(t)=\widehat{W}_n(t)$. Since for $t = {R}_n, {W}_n(t) =t $ then for $t = {R}_n, \widehat{W}_n(t) =t$, i.e., ${R}_n$ is a fixed point of the equation $\widehat{W}_n(t)=t$. Also, by definition, $\widehat{R}_n$ is the least fixed point of $\widehat{W}_n(t)=t$, hence we have $\widehat{R}_n \leq R_n$.  \hfill(**)

Now by \rle{le:exclusionIntervalRh}, the times of (*) are also contained in the permissible value range of $\widehat{R}_n$. Therefore for $t = \widehat{R}_n$, we also obtain $W_n(t)=\widehat{W}_n(t)$. Since for $t = \widehat{R}_n, \widehat{W}_n(t) =t $ then for $t = \widehat{R}_n, W_n(t) = t$, i.e., $\widehat{R}_n$ is a fixed point of the equation $W_n(t) =t$. Since $R_n$ is the least fixed point of $W_n(t)=t$, we have $R_n\leq\widehat{R}_n$. \hfill(***)

From (**) and (***), the Lemma is proved. 
%\qed
\endpro
We now assume that the activation times of the tasks $\tau_{\pi(i)}$ are delayed by a value 
$$0<\Delta_{\pi(i)} \leq C_{\pi(i+1)\ldots\pi(n-1)}$$\label{eq:del}
 and give an answer to the question whether this delays have an influence on the \emph{worst-case response time}.
\begin{lem}
Let be 
\begin{equation*}
R^{\Delta}_n=C_n+I^{\Delta}_{n-1}(R^{\Delta}_n)
\end{equation*}
with 
\begin{equation*}
I^{\Delta}_{n-1}(R^{\Delta}_n) = \sum_{i=1}^{n-1}C_{\pi(i)}\left\lceil (R^{\Delta}_n-\Delta_{\pi(i)})/T_{\pi(i)}\right\rceil
\end{equation*}
and
\begin{equation*}\label{eq:rel}
0 \leq \Delta_{\pi(i)} \leq C_{\pi(i+1)\ldots\pi(n-1)}
\end{equation*}
then
\begin{equation*}
R_n^{\Delta} = R_n
\end{equation*}
\end{lem}
\begin{proof}
Because of the relation \eqref{eq:rel} and Lemma , we get 
\begin{equation*}
\widehat{R}_n \leq R_n^{\Delta} \leq R_n
\end{equation*}
With Lemma \ref{le:RRh} holds
$$\widehat{R}_n = R_n^{\Delta} = R_n $$
%\qed
\end{proof}
\subsection{Strongly restricted different jitters for different tasks}\label{sec:restricted}
In this section, different values of the jitter are considered.  However, it is assumed that the task with the smallest period has the largest jitter and all other task jitters are appropriately limited so that the Lemmas of section \ref{sec:diff_models} are applicable. 

We have also to cope with the problem that jitter shifts the time of the discontinuities to the left whereas \emph{execution time intervals} go in the opposite direction. 

Note that  there is no margin for the value  $\Delta_{\pi(n-1)}$  in \eqref{eq:del} of the task with the smallest period i.e. we have $\Delta_{\pi(n-1)} = 0$.
To derive the ranges for the other jitters we assume $J_{max}=J_{\pi(n-1)}$ and consider
\begin{equation} \label{eq:IJmax}
W_{n-1}^{Jmax}(t)=C_n+\sum_{i=1}^{n-1}C_{\pi(i)}\left\lceil \frac{t+J_{\pi(n-1)}}{T_{\pi(i)}}\right\rceil
\end{equation}
and
\begin{equation} \label{eq:WJmax}
\widehat{W}_{n-1}^{Jmax}(t)=C_n+\sum_{i=1}^{n-1}C_{\pi(i)}\left\lceil \frac{t+J_{\pi(n-1)}-\Delta_{\pi(i)}}{T_{\pi(i)}}\right\rceil
\end{equation}
which lead to the same \emph{worst-case response time}.

Since the jitter $ J_{\pi(n-1)}$ in these equations is the same for all tasks we can determine the fixed point using our new method and get the \emph{worst-case response time}. 

We have shown in the previous section the extent to which the times of the discontinuities may vary without affecting the value of the fixed point. This possibility of variation does not change if the times of the discontinuities are shifted from $m_iT_{\pi(i)}$ to the new reference time $m_iT_{\pi(i)}-J_{\pi(n-1)}$. Accordingly, we introduce a virtual jitter into the formula for the \emph{total interference function} related to this reference time:
\begin{equation}\label{eq:IJi} 
I_{n-1}^{J_i}(t)=\sum_{i=1}^{n-1}C_{\pi(i)}\left\lceil \frac{t+J_{\pi(i)}}{T_{\pi(i)}}\right\rceil=
\sum_{i=1}^{n-1}C_{\pi(i)}\left\lceil \frac{t+J_{\pi(n-1)}-J_{\pi(n-1)}+J_{\pi(i)}}{T_{\pi(i)}}\right\rceil
\end{equation}
Where $t+J_{\pi(n-1)}$ describes the reference time and $\Delta_{\pi(i)}=J_{\pi(n-1)}-J_{\pi(i)}$ the virtual jitter bounded by \eqref{eq:del}.
\begin{lem}\label{lem:delta} 
For a set of $n$ tasks with jitters the equations
\begin{equation*}
R_n^{J_{max}} = W_{n-1}^{Jmax}(R_n^{J_{max}}) 
\end{equation*}
\begin{equation*}
R_n^{J_{i}} =\widehat{W}_{n-1}^{Jmax}(R_n^{J_{i}}) 
\end{equation*}
where the \emph{processor} demand function is defined in  \eqref{eq:IJmax} and \eqref{eq:WJmax} we get equal solutions:
\begin{equation*}
R_n^{J_{i}} = R_n^{J_{max}}
\end{equation*}
if 
\begin{equation}\label{eq:real_jitter_cond}
\forall 1 \leq i \leq n-1,\ \ \  max\left(0, J_{\pi(n-1)}-C_{\pi(i+1\dots\pi(n-1)}\right) \leq J_{\pi(i)} \leq J_{\pi(n-1)}
\end{equation}
\end{lem}
\begin{proof}
In order to fulfill the condition of Lemma \ref{lem:delta} for the the virtual jitter  $J_{\pi(n-1)}-J_{\pi(i)} = \Delta_{\pi(i)}$ must hold. Eq. \eqref{eq:rel} defines the range of  $\Delta_{\pi(i)}$. So we have
$0 \leq J_{\pi(n-1)}-J_{\pi(i)}  \leq  C_{\pi(i+1)\ldots\pi(n-1)}$. From this follows:
$$ J_{\pi(n-1)} - C_{\pi(i+1)\ldots\pi(n-1)}  \leq J_{\pi(i)} \leq J_{\pi(n-1)} $$
Since the upper bound could be $<0$ whereas jitters are $\geq 0$ we introduced the \textit{max}-function.
%\qed
\end{proof}

Since the maximum jitter should be assigned to the task $\tau_{n-1}$ we extend the rules for task ordering. Again,  $\pi$ describes a mapping of an ordered set from priority order to reverse rate monotonic ordering. A tie is broken by ordering the task with equal periods by growing jitter.

\begin{table}[htbp]
\caption{Example task set 1 }
\begin{center}
\begin{tabular}{|c|c|c|c|c|c|c|c|}

\hline
\multicolumn{4}{|c|}{\rule[-1mm]{0mm}{4.2mm}{\text{task parameters}}}&
\multicolumn{4}{|c|}{\rule[-1mm]{0mm}{4.2mm}{\text{ derived parameters} }}

\\
\hline 
\rule[-1mm]{0mm}{4.5mm}\textbf{$i$} & \textbf{$T_i$}& \textbf{$C_i$}&\textbf{$J_i$}& $U_i$&$T_i-J_i$&$J_{max}$&\textbf{response time}   \\
\hline
\rule[-0.5mm]{0mm}{4.1mm}1& 60 &6 & 8 &0.1&52&-&  6     \\
\hline
\rule[-0.5mm]{0mm}{4.1mm}2& 60 &8 & 0&0.133&60&8&  14     \\
\hline
\rule[-0.5mm]{0mm}{4.1mm}3& 30 &4 & 9 &0.133&21&9& 18   \\
\hline
\rule[-0.5mm]{0mm}{4.1mm}4& 360 &13 &7 &0.036&353&9& 35  \\
\hline
\rule[-0.5mm]{0mm}{4.1mm}5& 120 &7 &3 &0.058&117&9& 42  \\
\hline
\rule[-0.5mm]{0mm}{4.1mm}6& 360 &12 & 9 &0.033&351&9& 72     \\
\hline

\end{tabular} 
\label{tab1}
\end{center}
\end{table}

We now explain our method for determining the \emph{worst-case response time} using a concrete example whose parameters are listed in Table \ref{tab1}. The tasks are ordered by growing priority. We assume that the relative deadline of the tasks is $D_i = T_i$.
\begin{itemize}
	\item Task $\tau_1$: By \eqref{eq:bcwt} we have with $n=1$ and $C_{\pi(1)\dots\pi(n-1)}=0$ $R_1^{(0)}=C_1=6$.
	\item Task $\tau_2$: The only task with a higher priority is $\tau_1$. Thus we have $\pi: 1\rightarrow 1$. We get by \eqref{eq:bcwt} with $n=2$, $J=J_{max}=J_1=8$:  $\widetilde{R}_2^{(0)}=(C_2-J_{max})/(1-U_1) +J_{max}=88/9$. By \eqref{eq:itstwtJ}: $\widetilde{R}_2^{(1)}=\widetilde{R}_2^{(0)} -U_1(\widetilde{R}_2^{(0)}+J_{max})+C_1\left\lceil (\widetilde{R}_2^{(0)}+J_{max})/T_1 \right\rceil=14$
	\item Task $\tau_3$: We now have 2 tasks with higher priority, which have the same  periods. To sort by growing jitter we have to  reorder the two tasks. Thus we have  $\pi: 1,2 \rightarrow 2,1$ or $\tau_{\pi(1)}=\tau_2$ and $\tau_{\pi(2)}=\tau_1$.  
	
We get by \eqref{eq:bcwt} with $n=3$, $J=J_{max}=J_{\pi(2)}=8$:  

$\widetilde{R}_3^{(0)}=(C_3{-}J_{max})/(1{-}U_{\pi(1)\dots\pi(2)}) +J_{max}=176/23$. 
	
By \eqref{eq:itstwtJ}: 

$\widetilde{R}_3^{(1)} = \widetilde{R}_3^{(0)} + \left(-U_{\pi(1)}\left(\widetilde{R}_3^{(0)} + J_{max}\right) + C_{\pi(1)}\left\lceil(\widetilde{R}_3^{(0)}+J_{max}) / T_{\pi(1)} \right\rceil\right)$ / $\left(1-U_{\pi(2)}\right) = 128/9$
	
and $\widetilde{R}_3^{(2)}=\widetilde{R}_3^{(1)} -U_{\pi(2)}\left(\widetilde{R}_3^{(1)}+J_{max}\right)+C_{\pi(2)}\left\lceil (\widetilde{R}_3^{(1)}+J_{max})/T_{\pi(2)} \right\rceil = 18$.
\end{itemize}
	
The final results of the further steps are listed in Table \ref{tab1}. Note that for task with priority 4 we have $\pi: 1,2,3\rightarrow 2,1,3$, for task 5 $\pi:1,2,3,4\rightarrow 4,2,1,3$, and for task 6 $\pi:1,2,3,4,5\rightarrow 4,5,2,1,3$.  We see from the Table that all \emph{worst-case response times} $R_i$ are lower than $T_i-J_i$, thus the task system is schedulable.

% !TEX root = harmonic_27.11.2019.tex

 \subsection{Loosening the restrictions on jitter}\label{sec:looserestjit}
The condition introduced in the last section that $J_{\pi(n-1)}$ must be greater than any other jitter value $J_{\pi_i}$ with $1\leq i \leq n-2$, and the resulting other consequences for these jitters established in  \eqref{eq:real_jitter_cond}, may be too severe for practical applications. In this section we want to adapt the allowed jitter to larger value ranges.

 It is known that in harmonic systems jitters are often small and for some tasks even 0. So it makes sense, from a practical point of view, to limit the jitters for all $i$ to $J_i <T_i$. This does not mean, however, that we can specify a solution for all task systems that meet these constraints. Rather, we do allow a change in the restrictions, so that  the jitter of task $\tau_{\pi(n-1)}$ need not  be the largest one and that also several tasks can have a zero-jitter

For this purpose we introduce a new virtual jitter for the tasks

\begin{equation}\label{eq:vJ}
J'_{\pi(i)}= J_{\pi(i)} + m_i T_{\pi(i)}
\end{equation}
where for all $i$ $m_i \in \mathbb{N}_0$. 

 Introducing this virtual jitter into \eqref{eq:equalJ} we have to observe that the value of $I_{n-1}^{J_i}(t)$ for all $t$ is not changed. Therefore, outside of the ceiling functions, we subtract the same amount that we add up inside.
\begin{equation}\label{eq:Jmax2} 
\ \forall t \in \mathbb{R}, \ \ I_{n-1}^{J_i}(t)=\sum_{i=1}^{n-1}C_{\pi(i)}\left\lceil \frac{t{+}J_{\pi(i)}}{T_{\pi(i)}}\right\rceil=  \sum_{i=1}^{n-1}C_{\pi(i)}\left(\left\lceil \frac{t{+}J'_{\pi(i)}}{T_{\pi(i)}}\right\rceil-m_i\right)
\end{equation}
\textcolor[rgb]{0,0,0}{In section \ref{sec:restricted} we have shown that for real jitter values $J_{\pi(i)}$ that fulfill condition \eqref{eq:real_jitter_cond}, we can determine the \emph{worst-case response time} by replacing the constant jitter value $J$ with $J_{max}=J'_{\pi(n-1)}$ in Lemma \ref{thm2}. We now want to determine the virtual jitter values according to equation \eqref{eq:vJ} so that the same method can be applied to the virtual jitters this time by replacing $J$ with $J'_{max}=J'_{\pi(n-1)}$ that must be made greater than or equal to all other virtual jitters. Since the virtual jitter can be changed in steps of of height $T_{\pi(n-1)}$, there is a greater number of virtual jitter value sets for which the \emph{worst-case response time} can be determined using our method.}

%If the real jitter does not meet the conditions in Section 4.3, we may be able to set the virtual jitter so that we can determine the response time with the algorithm in Lemma 6 by replacing the Jitter $J$ which is the same for all tasks with $J'_{\pi(n-1)}$  if this jitter is greater than all other jitters.%   
%The jitter $J'_{\pi(n-1)}$ can be set in steps of height $T_{\pi(n-1)}$. %
The following therefore apply in detail:
\begin{enumerate}	
	\item \textcolor[rgb]{0,0,0}{The virtual jitter $J'_{\pi(n-1)}$ must be the largest virtual jitter i.e. $J'_{max}=_{def} \max_i J'_{\pi(i)}=J'_{\pi(n-1)}$ defines an upper bound  for all other virtual jitters: }
	\begin{equation}\label{eq:viJup}
	\forall i, \ \ 1\leq i \leq n-2 \ \ J'_{\pi(n-1)} \geq  J'_{\pi(i)}.
	\end{equation}
	\item In a similar way we can define a lower bound  for the virtual jitters
	\begin{equation}\label{eq:viJlo}
	\forall i, \ \ 1\leq i \leq n-2 \ \ J'_{\pi(n-1)}- C_{\pi(i+1)..\pi(n-1)} \leq  J'_{\pi(i)}.
	\end{equation}

	\item  We write \eqref{eq:viJup} and \eqref{eq:viJlo}
	
	in terms of  real task jitters. From \eqref{eq:viJup} follows
		\begin{equation}\label{eq:ineq2}
   \forall i,\ \ 1\leq i\leq n-2,\ \ J_{\pi(i)}+ m_iT_{\pi(i)} \leq J_{\pi(n-1)}+m_{n-1}T_{\pi(n-1)}
	\end{equation}

	If we consider that $J_{\pi(i)} \geq 0$ we get by \eqref{eq:viJlo} 
	\begin{multline}\label{eq:ineq}
	 \forall i,\ \ 1\leq i\leq n-2,\\ \max(m_iT_{\pi(i)}, J_{\pi(n-1)}+m_{n-1}T_{\pi(n-1)}- C_{\pi(i+1)..\pi(n-1)}) \leq  J_{\pi(i)}+ m_iT_{\pi(i)}
	\end{multline}

\end{enumerate}

Once we have found a valid set of values of the variables $m_i$, we can determine the \emph{total interference} as the basis of the fixed point iteration for determining the \emph{worst-case response time} as follows:
\begin{equation}\label{eq:interf}
I'^{\ \max}_{n-1}(t)=\sum_{i=1..n-1} C_{\pi(i)}\left(\left\lceil \frac{t{+}J'_{max}}{T_{\pi(i)}}\right\rceil-m_i\right)
\end{equation}
Jitter is now constant in all ceiling terms, so we can use the Lemma 6 \ref{thm2} where we set $J=J'_{max}$ to determine the fixed point for the recursive equation:
\begin{equation}\label{eq:resp_time}
R_{n}= C_{n} + I'^{\ \max}_{n-1}(R_{n})
\end{equation}

The variables $m_i$ must  be integers and therefore we have a special type of constraint programming problem. The set of points that satisfy the constraints  is called feasibility region.  If there are no such points  the feasible region is the null set and the problem has no solution what means that it is infeasible. If there exists at least one solution the constraint program is feasible.

Note that we  have two variables per inequality and  the variable $m_{n-1 }$ is contained in every constraints. Furthermore we have two constraints for each  pair of variables $m_{n-1 }$ and $m_i$ with $1 \leq i \leq {n-2}$ and therefore there are $2(n-2)$ constraints. Our  system of constraints is called monotone since each constraint is an inequality on two variables with coefficients of opposite signs. 

The property of two variables per constraint present in our problem type has been extensively discussed in the literature, generally assuming that bounds are known for the value ranges of the variables. In \cite{Bar-Yehuda2001} more general integer programs with 2 variables per constraint are considered and an $\mathcal{O}(2\Delta(n-2))$  feasibility algorithm is proposed, where $\Delta$ denotes the maximum value range of any variable.
As proved in \cite{Hochbaum:1994:SFA:194602.194606}  the problem of finding a feasible solution of a system of monotone inequalities in integers is weak NP complete. The proposed algorithm transforms a fractional solution of the corresponding LP program step by step into a solution of the ILP. 
  
In our problem, such bounds are not present as values, but are represented by terms containing the variable $m_{n-1}$. In addition, all feasible solutions are equivalent, since we can apply appropriate corrections to the \emph{total interference function} \eqref{eq:Jmax2}.

\subsubsection{Defining the value of $m_1$}
If the constraint program is feasible we can calculate from a valid solution infinitely many other solutions by adding to or subtracting from all $m_i$ values a multiple of the integer $T_{\pi(1)}/T_{\pi(i)}$ i.e.,  $\alpha T_{\pi(1)}/T_{\pi(i)}$ where $\alpha \in \mathbb{Z}$. This is shown with the following Lemma.
\begin{lem}
If we know a solution in the feasibility region of the constraint program described in \eqref{eq:ineq2} and \eqref{eq:ineq} with the variable values $\widehat{m}_i$ then there are also infinity many solutions in this region with the  values  of the variables
\begin{equation}\label{eq:mod}
\forall 1 \leq i \leq n-1, \ \ m_{i} = \widehat{m}_{i}+ \alpha \frac{T_{\pi(1)}}{T_{\pi(i)}}
\end{equation}
where $\alpha \in \mathbb{Z}$.
\end{lem}
\begin{proof}
Any solution satisfies the constraints in \eqref{eq:ineq2} and \eqref{eq:ineq} which define a lower and an upper bound, respectively. For the assumed solution the difference $	\widehat{m}_{n-1}T_{\pi(n-1)}- \widehat{m}_iT_{\pi(i)}$ must lay within these bounds. We have to show that any proposed transformation does not change the value of this difference.
By \eqref{eq:mod}  we  get:
	\begin{multline}\label{eq:ineq8}
 \widehat{m}_{n-1}T_{\pi(n-1)}- \widehat{m}_iT_{\pi(i)} =	m_{n-1}T_{\pi(n-1)}-\alpha T_{\pi(1)}- m_iT_{\pi(i)} +\alpha T_{\pi(1)}= \\	m_{n-1}T_{\pi(n-1)}- m_iT_{\pi(i)} 
	\end{multline}
The new difference is the same as the old one and therefore also stays within the limits. Hence the lemma follows.	
\qed	
\end{proof}
\textcolor[rgb]{0,0,0}{Note that in \eqref{eq:Jmax2} the values  of the function  $I^{J_i}_{n-1}(t)$ over $t$ do not change for any of the possible solution set ${m_1,\dots,m_{n-1}}$. }Larger values of $m_i$ result in  larger values of virtual jitters and a larger reduction outside the ceiling terms.   So we can select any of the solutions as a representative. In the following we set $$m_{1} = 1$$ with the consequence that all $m_i \geq 0$
  which is demonstrated by verifying the validity of the following Lemma.
\begin{lem}
We consider the constraint program established by \eqref{eq:ineq2} and \eqref{eq:ineq}. For $m_1=1$ we get $m_i \geq 0$ with $1\leq i \leq n-1$
\end{lem} 
\begin{proof}
By \eqref{eq:ineq} we have
	\begin{equation*}
	 \forall i,\ \ 1\leq i\leq n-2, J_{\pi(n-1)}+m_{n-1}T_{\pi(n-1)}- C_{\pi(i+1)\dots\pi(n-1)} \leq  J_{\pi(i)}+ m_iT_{\pi(i)}
	\end{equation*}
	and by \eqref{eq:ineq2} we have with $i=1$ and $m_1=1$:
	\begin{equation}
	J_{\pi(n-1)}+m_{n-1}T_{\pi(n-1)}\geq J_{\pi(1)}+ T_{\pi(1)} 
	\end{equation}
	We combine the two constraints and resolve to $m_i$. Additionally we consider the integer of the division result $T_{n-1}|T_i$ and $m_i$:
	\begin{equation}
		 \forall i,\ \ 1\leq i\leq n-2, \  \  m_i \geq {T_{\pi(1)} \over T_{\pi(i)}}+\left\lceil {J_{\pi(1)}-J_{\pi(i)}-C_{\pi(i+1)\dots\pi(n-1)}\over T_{\pi(i)} }\right\rceil
	\end{equation}
	With $J_{\pi(1)} \geq 0$, $J_{\pi(i)} < T_{\pi(i)}$, and $C_{\pi(i+1)\dots\pi(n-1)} < T_{\pi(i)} $ the second summand is $\geq -1$ whereas the first summand is $\geq 1$. It follows for all $i$ $m_i \geq 0$.
	%and the derived lower bound $m_{i.lo}$ for $m_i$
	\qed
\end{proof}

In Section 4.3, we have proposed a method where it was implicitly assumed that all $m_i=0$ with $1 \leq i \leq n-1$. Such a solution is now impossible, because we have explicitly set $m_1=1$.  If we have jitter values that are valid according to Section 4.3, there are also valid $m_i$ values according to Lemma 1 under the assumption $m_1=1$.

\iffalse In section \ref{} we have discussed a solution for $m_i=0$ with $1 \leq i \leq n-1$. This solution is now excluded since we have defined $m_1=1$. Nevertheless we get a solution for all real jitters which are feasible with $m_i=0$ for all $i$ even with the formulas in this section, but now with $m_iT_{\pi(i)}=T_{\pi(1)}$ for all $i$. If this equality is maintained, we can subtract $m_iT_{\pi(i)}$ or $m_{n-1}T_{\pi(n-1)} $ in \eqref{eq:ineq2} and \eqref{eq:ineq} on both sides and obtain the constraints from section \ref{}. 
\fi

\subsection{Derivation of an algorithm}\label{constprog}
The task of the algorithm presented below is to determine the $m_i$ values of a feasible constraint system or to characterize the system as infeasible. Since we also want  to consider  the efficiency of the algorithm, in certain cases we use a simple heuristic that has to choose between two possible values.

The basic approach is to determine the values $m_2, m_3, \dotsm_{n-1}$ one after the other starting from the fixed value $m_1=1$. Therefore we start again with the constraints \eqref{eq:ineq2} and \eqref{eq:ineq}, select as index $i$ and $i+1$, and set $m_i$ and $m_{i+1}$ in relation to each other.

\begin{equation}\label{eq:constraint_set1}
\forall i, \ \ 1\leq i \leq n-2, \ \ m_{i+1}T_{\pi(i+1)} + J_{\pi(i+1)}+C_{\pi(i+2)\dots\pi(n-1)} \geq m_{i}T_{\pi(i)} + J_{\pi(i)}
\end{equation}
\begin{equation}\label{eq:constraint_set2}
m_{i+1}T_{\pi(i+1)} + J_{\pi(i+1)} \leq m_{i}T_{\pi(i)} + J_{\pi(i)}+C_{\pi(i+1)\dots\pi(n-1)}
\end{equation}
Note that $C_{\pi(n)\dots\pi(n-1)}=0$. 
We rearrange terms, exploit that $m_{i+1}$ must be an integer, and that $T_{\pi(i+1)}$ divides the period $T_{\pi(i)}$.
\begin{multline}\label{eq:m.i+1}
{T_{\pi(i)} \over T_{\pi(i+1)}}m_i {+} \left\lceil {J_{\pi(i)} {-} J_{\pi(i+1)} {-} C_{\pi(i{+}2)\dots\pi(n{-}1)} \over T_{\pi(i+1)}}\right\rceil \leq m_{i+1} \leq \\  {T_{\pi(i)} \over T_{\pi(i+1)}}m_i {+} \left\lfloor {J_{\pi(i)} {-} J_{\pi(i+1)} {+} C_{\pi(i+1)\dots\pi(n-1)} \over T_{\pi(i+1)}}\right\rfloor 
\end{multline}
Let us first check under which conditions which values for $m_{i+1}$ are allowed, assuming that the value of $m_i$ is unique or has been chosen by a heuristic technique.

In order to keep the presentation clear, we introduce the abbreviation $\widetilde{J}_{\pi(i)} =_{def} J_{\pi(i+1)}-J_{\pi(i)} \mod T_{\pi(i+1)}$ and use the universally valid identity $x \mod y = x - y\left\lfloor x/y\right\rfloor$. We can therefore transform \eqref{eq:m.i+1} into:
\begin{multline}\label{eq:unique}
{T_{\pi(i)} \over T_{\pi(i+1)}}m_i {+}\left\lfloor {J_{\pi(i)}-J_{\pi(i+1)} \over T_{\pi(i+1)}} \right\rfloor +\left\lceil {\widetilde{J}_{\pi(i)} {-} C_{\pi(i{+}2)\dots\pi(n{-}1)} \over T_{\pi(i+1)}}\right\rceil \leq m_{i+1} \leq \\  {T_{\pi(i)} \over T_{\pi(i+1)}}m_i {+}\left\lfloor {J_{\pi(i)}-J_{\pi(i+1)} \over T_{\pi(i+1)}} \right\rfloor +\left\lfloor {\widetilde{J}_{\pi(i)} {+} C_{\pi(i+1)\dots\pi(n-1)} \over T_{\pi(i+1)}}\right\rfloor 
\end{multline}
where $0 \leq \widetilde{J}_{\pi(i)} /T_{\pi(i+1)}<1$ and $C_{\pi(i{+}2)\dots\pi(n{-}1)}/T_{\pi(i+1)}< C_{\pi(i{+}1)\dots\pi(n{-}1)}/T_{\pi(i+1)}<1$, i.e.
$$\left\lceil {\widetilde{J}_{\pi(i)} {-} C_{\pi(i{+}2)\dots\pi(n{-}1)} \over T_{\pi(i+1)}}\right\rceil \in \left\{0,1\right\}$$
and
$$\left\lfloor {\widetilde{J}_{\pi(i)} {+} C_{\pi(i+1)\dots\pi(n-1)} \over T_{\pi(i+1)}}\right\rfloor \in \left\{0,1\right\}$$
We define $\gamma_i \in \left\{\left\lceil {\widetilde{J}_{\pi(i)} {-} C_{\pi(i{+}2)\dots\pi(n{-}1)} \over T_{\pi(i+1)}}\right\rceil, \left\lfloor {\widetilde{J}_{\pi(i)} {+} C_{\pi(i+1)\dots\pi(n-1)} \over T_{\pi(i+1)}}\right\rfloor \right\}$ and discuss the possible situations depending on the value range of $\gamma_i$ in the following Lemma.

Note that if $m_i$ has a unique value and the two terms defining the value range of $\gamma_i$ are equal, $m_{i+1}$ is unique.
\begin{lem}\label{lem:gamma}
We have $\widetilde{J}_{\pi(i)}<T_{\pi(i+1)}$ and $C_{\pi(i+2)\dots\pi(n-1)} < C_{\pi(i+1)\dots\pi(n-1)} <T_{\pi(i+1)}$ then \eqref{eq:unique} has the following possible solutions depending on $\widetilde{J}_{\pi(i)}:\newline$
\begin{equation}\label{eq:gamma}
\gamma_i \in \begin{cases}\left\{0\right\}&\widetilde{J}_{\pi(i)} \leq \min( C_{\pi(i+2)\dots\pi(n-1)},T_{\pi(i+1)}-C_{\pi(i+1)\dots\pi(n-1)}-1) 
\\\left\{1\right\}&\max( C_{\pi(i+2)\dots\pi(n-1)}+1,T_{\pi(i+1)}-C_{\pi(i+1)\dots\pi(n-1)}) \leq \widetilde{J}_{\pi(i)} \\ \emptyset&C_{\pi(i+2)\dots\pi(n-1)}+1\leq \widetilde{J}_{\pi(i)} \leq T_{\pi(i+1)}-C_{\pi(i+1)\dots\pi(n-1)}-1\\
\left\{0,1\right\}&T_{\pi(i+1)}-C_{\pi(i+1)\dots\pi(n-1)} \leq \widetilde{J}_{\pi(i)}\leq C_{\pi(i+2)\dots\pi(n-1)}
\end{cases}
\end{equation}
\end{lem}
\begin{proof}
We look on the 4 cases:
\begin{enumerate}
\item
We have $$\left\lceil \frac{\widetilde{J}_{\pi(i)}-C_{\pi(i+2)\dots\pi(n-1)}}{T_{\pi(i+1)}}\right\rceil = 0 \Leftrightarrow -1 < \frac{\widetilde{J}_{\pi(i)}-C_{\pi(i+2)\dots\pi(n-1)}}{T_{\pi(i+1)}} \leq 0$$
We dissolve to $\widetilde{J}_{\pi(i)}$ and get $C_{\pi(i+2)\dots\pi(n-1)}-T_{\pi(i+1)} < \widetilde{J}_{\pi(i)}\leq C_{\pi(i+2)\dots\pi(n-1)}$. The LHS is lower than 0 but $\widetilde{J}_{\pi(i)} \geq 0$ by definition. We therefore can write $0 \leq \widetilde{J}_{\pi(i)}\leq C_{\pi(i+2)\dots\pi(n-1)}$

We perform a similar consideration for the floor function 
$$\left\lfloor \frac{\widetilde{J}_{\pi(i)}+C_{\pi(i+1)\dots\pi(n-1)}}{T_{\pi(i+1)}}\right\rfloor = 0 \Leftrightarrow 0 \leq \frac{\widetilde{J}_{\pi(i)}+C_{\pi(i+1)\dots\pi(n-1)}}{T_{\pi(i+1)}} < 1$$
and get $0 \leq \widetilde{J}_{\pi(i)} < T_{\pi(i+1)} - C_{\pi(i+1)\dots\pi(n-1)}$. Since periods and worst-case execution times are integers we have $0 {\leq} \widetilde{J}_{\pi(i)} {\leq} T_{\pi(i+1)} {-} C_{\pi(i+1)\dots\pi(n-1)}-1$. 

The ceiling term and the floor term must both have the value 0, so that the minimum of the two upper limits must apply.
\item
In this case, both terms must have the value 1.
$$\left\lceil \frac{\widetilde{J}_{\pi(i)}-C_{\pi(i+2)\dots\pi(n-1)}}{T_{\pi(i+1)}}\right\rceil = 1 \Leftrightarrow 0 < \frac{\widetilde{J}_{\pi(i)}-C_{\pi(i+2)\dots\pi(n-1)}}{T_{\pi(i+1)}} \leq 1$$ and
$$\left\lfloor \frac{\widetilde{J}_{\pi(i)}+C_{\pi(i+1)\dots\pi(n-1)}}{T_{\pi(i+1)}}\right\rfloor = 1 \Leftrightarrow 1 \leq \frac{\widetilde{J}_{\pi(i)}+C_{\pi(i+1)\dots\pi(n-1)}}{T_{\pi(i+1)}} < 2$$
This time we need to look at the lower bounds and make the maximum.
$\widetilde{J}_{\pi(i)} \geq \max(C_{\pi(i+2)\dots\pi(n-1)}+1, T_{\pi(i+1)}- C_{\pi(i+1)\dots\pi(n-1)})$
\item
In this case, the lower limit is greater than the upper limit, i.e. the ceiling term assumes the value 1 and the floor term the value 0. This requires compliance with the constraints:
$ 0 \leq \frac{\widetilde{J}_{\pi(i)}+C_{\pi(i+1)\dots\pi(n-1)}}{T_{\pi(i+1)}} < 1 \wedge 0 < \frac{\widetilde{J}_{\pi(i)}-C_{\pi(i+2)\dots\pi(n-1)}}{T_{\pi(i+1)}} \leq 1 $
We take the maximum of the lower limits and the minimum of the upper limits and get:
$$C_{\pi(i+2)\dots\pi(n-1)}+1 \leq \widetilde{J}_{\pi(i)}\leq T_{\pi(i+1)}- C_{\pi(i+1)\dots\pi(n-1)}-1$$
\item
This case is characterized by the fact that the ceiling term assumes the value 0 and the floor term the value 1, i.e. $ -1 < \frac{\widetilde{J}_{\pi(i)}-C_{\pi(i+2)\dots\pi(n-1)}}{T_{\pi(i+1)}} \leq 0 \wedge 1 \leq \frac{\widetilde{J}_{\pi(i)}+C_{\pi(i+1)\dots\pi(n-1)}}{T_{\pi(i+1)}} < 2$. The maximum of the lower limits and the minimum of the upper limits leads to the constraints: $T_{\pi(i+1)}- C_{\pi(i+1)\dots\pi(n-1)} \leq \widetilde{J}_{\pi(i)} \leq C_{\pi(i+2)\dots\pi(n-1)}$.
\end{enumerate}
\qed
\end{proof}
We can derive an interesting special case from Lemma \ref{lem:gamma}.
\begin{cor}\label{cor:special_case}
If $C_{\pi(i+1)\dots\pi(n-1)} +C_{\pi(i+2)\dots\pi(n-1)} < T_{\pi(i+1)}$ then $m_i$ is unique or the system is infeasible.
\end{cor}
\begin{proof}
The specified restriction excludes case 4 i.e.
\begin{multline*}
T_{\pi(i+1)}-C_{\pi(i+1)\dots\pi(n-1)} \leq \widetilde{J}_{\pi(i)}\leq C_{\pi(i+2)\dots\pi(n-1)} \Rightarrow \\
T_{\pi(i+1)}-C_{\pi(i+1)\dots\pi(n-1)} \leq C_{\pi(i+2)\dots\pi(n-1)}  
\end{multline*}
\qed
\end{proof}

With Lemma \ref{lem:gamma} it is clear that with a fixed value of $m_i$, the value range of $m_{i+1}$ comprises at most 2 values. 

If we compare the original constraint set defined in \eqref{eq:ineq2} and \eqref{eq:ineq} with that of \eqref{eq:constraint_set1} and \eqref{eq:constraint_set2}, we realize that the restrictions of each virtual jitter by the virtual jitter of task $\tau_{\pi(n-1)}$ has been lost. This can cause the calculated $m_i$ values not to comply with these decisive constraints. To avoid this, we calculate the value range for $m_{n-1}$ after each determination of an $m_i$ value. If it is empty, then there is no feasible solution for the constraint problem.
Based on the definition of $m_1 =1$ we get by \eqref{eq:ineq2} and \eqref{eq:ineq}:
\begin{equation}\label{eq:task1}
{T_{\pi(1)} \over  T_{\pi(n-1)}}  +\left\lceil {J_{\pi(1)}{-}J_{\pi(n-1)} \over T_{\pi(n-1)}} \right\rceil \leq m_{n-1} \leq {T_{\pi(1)} \over  T_{\pi(n-1)}}  +\left\lfloor {J_{\pi(1)}{-}J_{\pi(n-1)} {+} C_{\pi(2)\dots\pi(n-1)} \over T_{\pi(n-1)}} \right\rfloor
\end{equation}
If a value $m_i$ for $i > 1$ is fixed, then the value range for $m_{n-1}$ can be defined as follows:
\begin{equation}
{m_iT_{\pi(i)} \over  T_{\pi(n-1)}}  +\left\lceil {J_{\pi(i)}{-}J_{\pi(n-1)} \over T_{\pi(n-1)}} \right\rceil {\leq} m_{n-1} {\leq} {m_iT_{\pi(i)} \over  T_{\pi(n-1)}}  +\left\lfloor {J_{\pi(i)}{-}J_{\pi(n-1)} {+} C_{\pi(i+1)\dots\pi(n-1)} \over T_{\pi(n-1)}} \right\rfloor
\end{equation} 
Since the lower and upper limits do not grow or sink monotonously with growing i, we compute the maximum of the lower limits and the minimum of the upper limits in order to see whether there is still an admissible value for $m_{n-1}$.
\begin{multline}\label{m.i+1.2}
m_{n-1.lb.i}=_{def}\max_{j=1\dotsi}\left( {m_j T_{\pi(j)} \over T_{\pi(n-1)}}+\left\lceil {J_{\pi(j)}-J_{\pi(n-1)}\over T_{\pi(n-1)}}\right\rceil\right)\leq m_{n-1} \leq \\
\min_{j=1\dotsi}\left( {m_j T_{\pi(j)} \over T_{\pi(n-1)}}+\left\lfloor {J_{\pi(j)}-J_{\pi(n-1)}+C_{\pi(j+1)\dots\pi(n-1)}\over T_{\pi(n-1)}}\right\rfloor\right)=_{def} m_{n-1.ub.i}
\end{multline}
If $m_{n-1.lb.i}>m_{n-1.ub.i}$, then the constraint system is infeasible.
For an iterative calculation of $m_{n-1.lb.i}$ we can also write
\begin{equation}\label{eq:lb.i+1}
m_{n-1.lb.i+1}=\max\left(m_{n-1.lb.i}, {m_{i+1} T_{\pi(i+1)} \over T_{\pi(n-1)}}+\left\lceil {J_{\pi(i+1)}-J_{\pi(n-1)}\over T_{\pi(n-1)}}\right\rceil\right)
\end{equation}
and
\begin{multline}\label{eq:ub.i+1}
m_{n-1.ub.i+1} = \\ \min\left(m_{n-1.lb.i}, {m_{i+1} T_{\pi(i+1)} \over T_{\pi(n-1)}}{+}\left\lfloor {J_{\pi(i+1)}{-}J_{\pi(n-1)}{+}C_{\pi(i+2)\dots\pi(n-1)}\over T_{\pi(n-1)}}\right\rfloor\right)
\end{multline}

We can use the limitations for $m_{n-1}$ in \eqref{m.i+1.2} to get new restrictions for $m_{i+1}$ which follow from merging \eqref{eq:ineq2} and \eqref{eq:ineq} into a lower than or equal to chain:
\begin{multline*}
\left\lceil {m_{n-1}T_{\pi(n-1)}+J_{\pi(n-1)}-J_{\pi(i+1)}-C_{\pi(i+2)\dots\pi(n-1)}\over T_{\pi(i+1)}}\right\rceil \leq m_{i+1}\leq \\ \left\lfloor {m_{n-1}T_{\pi(n-1)}+J_{\pi(n-1)}-J_{\pi(i+1)}\over T_{\pi(i+1)}}    \right\rfloor
\end{multline*}

We replace $m_{n-1}$ in the LHS of the inequality by $m_{n-1.lb.i}$ and in the RHS by $m_{n-1.ub.i}$.
\begin{equation*}
{\left\lceil{{\max\limits_{j{=}1\dotsi}}\left({m_j T_{\pi(j)}}{+}T_{\pi(n{-}1)}{\left\lceil {J_{\pi(j)}{-}J_{\pi(n{-}1)}\over T_{\pi(n{-}1)}}\right\rceil}\right){+}J_{\pi(n{-}1)}{-}J_{\pi(i{+}1)}{-}C_{\pi(i{+}2)\dots\pi(n{-}1)}\over T_{\pi(i{+}1)}}\right\rceil} 
\end{equation*}
\begin{equation}\label{eq:m.i+1.2}
 \leq m_{i{+}1}   \leq
\end{equation}
\begin{equation*}
  \left\lfloor {\min\limits_{j=1\dotsi}\left( m_j T_{\pi(j)} {+}T_{\pi(n-1)}\left\lfloor {J_{\pi(j)}-J_{\pi(n-1)}+C_{\pi(j+1)\dots\pi(n-1)}\over T_{\pi(n-1)}}\right\rfloor\right){+}J_{\pi(n-1)}{-}J_{\pi(i+1)}\over T_{\pi(i+1)}}\right\rfloor
\end{equation*}
We compare the limits for $m_{i+1}$ in \eqref{eq:m.i+1} and \eqref{eq:m.i+1.2}, which were determined in different ways. For the lower limits, we take the index $i$ in \eqref{eq:m.i+1.2}, which is contained in the index set over which the maximum is to be taken.  We also observe that $\left\lceil x\right\rceil \geq x$ and perform some simplifications. Finally, we use the definition introduced in \eqref{m.i+1.2} to make the presentation clearer.
\begin{multline}\label{eq:lower_bound}
m_{i+1} \geq m_{i+1.lb} \equals \left\lceil {m_{n-1.lb.i}T_{\pi(n-1)}{+}J_{\pi(n-1)}{-}J_{\pi(i+1)}{-}C_{\pi(i{+}2)\dots\pi(n{-}1)}\over T_{\pi(i{+}1)}}\right\rceil \\ \geq 
\left\lceil { {m_i T_{\pi(i)} } {+}J_{\pi(i)}{-}J_{\pi(i+1)}{-}C_{\pi(i{+}2)\dots\pi(n{-}1)}\over T_{\pi(i{+}1)}}\right\rceil 
\end{multline}
The lower limit of equation \eqref{eq:m.i+1} is therefore lower than or equal to  the lower limit of \eqref{eq:m.i+1.2}. In a similar way, we compare the upper limits in \eqref{eq:m.i+1} and \eqref{eq:m.i+1.2} and get:
\begin{multline}\label{eq:m.i+1.3}
m_{i+1}{\leq}m_{i+1.ub} \equals \left\lfloor {m_{n-1.ub.i}}T_{\pi(n-1)}{+}J_{\pi(n-1)}{-}J_{\pi(i+1)}\over T_{\pi(i+1)}    \right\rfloor \\ \leq \left\lfloor {m_i T_{\pi(i)} {+} J_{\pi(i)}{+}C_{\pi(i+1)\dots\pi(n-1)}{-}J_{\pi(i+1)}\over T_{\pi(i+1)}} \right\rfloor
\end{multline}
The upper and lower limits of equation \eqref{eq:m.i+1.2} are therefore stricter than those of equation \eqref{eq:m.i+1}. 

A special situation is for $i+1=n-1$. Then we have by definition with $C_{\pi(i+2)\dots\pi(n-1)}=C_{\pi(n)\dots\pi(n-1)}=0$ and $J_{\pi(i+1)}=J_{\pi(n-1)}$ the lower bound $m_{n-1.lb}=m_{n-1.lb.n-2}$ and the upper bound $m_{n-1.ub}=m_{n-1.ub.n-2}$. If $m_{n-1.lb.n-2} \leq m_{n-1.ub.n-2}$ we choose $m_{n-1}=m_{n-1.lb.n-2}$.

Both in \eqref{m.i+1.2} and in  \eqref{eq:lower_bound}, \eqref{eq:m.i+1.3}, it must be ensured that the respective lower limit is less than or equal to the corresponding upper limit, so that we obtain valid values for $m_{n-1}$ and $m_{i+1}$.
\begin{lem}
If for all $i$ the virtual jitters meet the constraints:
\begin{multline}\label{eq:uniqueness}
\widetilde{J}_{\pi(i)} \leq \min( C_{\pi(i+2)\dots\pi(n-1)},T_{\pi(i+1)}-C_{\pi(i+1)\dots\pi(n-1)}-1) \ \ \vee \\ \widetilde{J}_{\pi(i)}
\geq \max( C_{\pi(i+2)\dots\pi(n-1)}+1,T_{\pi(i+1)}-C_{\pi(i+1)\dots\pi(n-1)})  
\end{multline}
then the constraint system  has a unique solution or is infeasible. 
\end{lem}
\begin{proof}
We assume $m_{n-1.lb.i} \leq m_{n-1.ub.i}$ and $m_{i+1.lb}\leq m_{i+1.ub}$, otherwise the system is infeasible.
By the first two cases of \eqref{eq:gamma} the constraints in \eqref{eq:uniqueness} for the virtual jitters  lead to a unique solution of $m_{i+1}$ for a unique $m_i$. This means that for all i
the upper limit and the lower limit for $m_{i+1}$ in \eqref{eq:m.i+1} are equal.
We join the two inequalities \eqref{eq:lower_bound} and \eqref{eq:m.i+1.3} together:
\begin{multline*}\label{eq:m.i+1.4}
\left\lceil { m_i T_{\pi(i)}  {+}J_{\pi(i)}{-}J_{\pi(i+1)}{-}C_{\pi(i{+}2)\dots\pi(n{-}1)}\over T_{\pi(i{+}1)}}\right\rceil \leq\\   \left\lceil {m_{n-1.lb.i}}T_{\pi(n-1)}{+}J_{\pi(n-1)}{-}J_{\pi(i+1)}{-}C_{\pi(i{+}2)\dots\pi(n{-}1)}\over T_{\pi(i{+}1)}\right\rceil \leq 
m_{i+1} \leq \\ \left\lfloor {m_{n-1.ub.i}}T_{\pi(n-1)}{+}J_{\pi(n-1)}{-}J_{\pi(i+1)}\over T_{\pi(i+1)}    \right\rfloor  \leq \\ \left\lfloor {m_i T_{\pi(i)} {+} J_{\pi(i)}{+}C_{\pi(i{+}1)\dots\pi(n{-}1)}{-}J_{\pi(i{+}1)}\over T_{\pi(i+1)}} \right\rfloor
\end{multline*}
The inner terms need valid values of $m_{n-1.lb.i}$ and $m_{n-1.ub.i}$ i.e.  $m_{n-1.lb.i}\leq m_{n-1.ub.i}$  and a valid value $m_{i+1}$ i.e., $m_{i+1.lb}\leq m_{i+1.ub}$. Otherwise the system is infeasible.
Since the two outer terms are equal if the virtual jitters meet the restrictions mentioned above, the values of all terms in the chain must be equal in case of a feasible system and define exactly one value of $m_i$ for all $1\leq i \leq n-2$.
\qed
\end{proof}

An interesting consequence for the further course of the calculations arises in the case that $m_{n-1.lb.i}=m_{n-1.ub.i}$ applies. We show this in the following Corollary.
\begin{cor}
If $m_{n-1.lb.i}=m_{n-1.ub.i}$  for some $i$, then  $m_{n-1.lb.j}=m_{n-1.ub.j}$ and $m_{j.lb} = m_{j.ub}$ for any $i<j\leq n-1$ or the system is infeasible. 
\end{cor}
\begin{proof}
 In \eqref{eq:m.i+1} we have defined $m_{i+1.ub} = \left\lfloor {m_{n-1.ub.i}}T_{\pi(n-1)}{+}J_{\pi(n-1)}{-}J_{\pi(i+1)}\over T_{\pi(i+1)}    \right\rfloor$ and in \eqref{eq:lower_bound}  $m_{i+1.lb} = \left\lceil m_{n-1.lb.i}T_{\pi(n-1)}{+}J_{\pi(n-1)}{-}J_{\pi(i+1)}{-}C_{\pi(i{+}2)\dots\pi(n{-}1)}\over T_{\pi(i{+}1)}\right\rceil$. 
%For $m_{n-1.lb.i}=m_{n-1.ub.i}$  the argument of the ceiling function is $\frac{C_{\pi(i{+}2 }}{ T_{\pi(i+1)}} < 1 $ smaller than the argument of the floor function. 

We denote the argument of the floor function by $\nu_{i+1}+\epsilon_{i+1}$ with $\nu_{i+1} \in \mathbb{N}$ and $0 \leq \epsilon_{i+1} < 1$. Then $m_{i+1.ub} = \nu_{i+1}$.
The lower bound is now
$$m_{i+1.lb}=\left\lceil \nu_{i+1} +\epsilon_{i+1} -\delta_{i+1}\right\rceil $$
where $\delta_{i+1} = C_{\pi(i+2)\dots\pi(n-1)} /T_{\pi(i+1)}$ which is lower than $U_{\pi(i+1)\dots\pi(n-1)}$ and therefore $< 1$. Evaluating the argument of the ceiling function we get

 \begin{equation}\label{cases} 
m_{i+1.lb} = \begin{cases} \nu_{i+1}\ \  \text{if}\ \  \delta_{i+1} \geq \epsilon_{i+1} \\ \nu_{i+1}+1 \ \  \text{if}\ \  \delta_{i+1} < \epsilon_{i+1}  \end{cases}
\end{equation}

Therefore we get $m_{i+1.lb} \geq m_{i+1.ub}$ and we only have a feasible system for equality.

Considering Eq. \eqref{eq:lb.i+1} and \eqref{eq:ub.i+1} and observing $m_{n-1.lb.i}=m_{n-1.ub.i}$ we have to analyze the second arguments of the $\min$- and the $\max$-function. In order to leave at least one valid value for $m_{n-1}$ it must be:
\begin{multline*}
 {m_{i+1} T_{\pi(i+1)} \over T_{\pi(n-1)}}+\left\lceil {J_{\pi(i+1)}-J_{\pi(n-1)}\over T_{\pi(n-1)}}\right\rceil \leq m_{n-1.lb.i} = m_{n-1.ub.i} \leq \\ {m_{i+1} T_{\pi(i+1)} \over T_{\pi(n-1)}}{+}\left\lfloor {J_{\pi(i+1)}{-}J_{\pi(n-1)}{+}C_{\pi(i+2)\dots\pi(n-1)}\over T_{\pi(n-1)}}\right\rfloor
\end{multline*}
The maximum over the term  $m_{n-1.lb.i}$  and the leftmost term of the inequalities above yields $m_{n-1.lb.i+1}$, whereas the minimum over the terms of the right inequality yields $m_{n-1.ub.i}$.
Therefore from $m_{n-1.lb.i} = m_{n-1.ub.i}$ follows $m_{n-1.lb.i+1}=m_{n-1.lb.i}= m_{n-1.ub.i} = m_{n-1.ub.i+1}$ or the system is infeasible.
\qed
\end{proof}
 
We now want to derive further conditions under which the constraint system always has a unique solution. Note that for $C_{\pi(i{+}1)\dots\pi(n{-}1)}+C_{\pi(i{+}2)\dots\pi(n{-}1)}< T_{\pi(i+1)}$ the 4th case cannot occur in \eqref{eq:gamma} and the constraint program either has a unique solution or is not feasible (see Corollary  \ref{cor:special_case}). It is 
\begin{equation*}
{C_{\pi(i{+}1)\dots\pi(n{-}1)}\over T_{\pi(i+1)}} = \sum_{j=i+1}^{n-1} U_{\pi(j)}{T_{\pi(j)}\over T_{\pi(i+1)}}\leq \sum_{j=i+1}^{n-1} U_{\pi(j)}\leq 1 - \sum_{j=1}^{i}U_{\pi(j)}
\end{equation*}
Accordingly we have
\begin{equation*}
{C_{\pi(i{+}2)\dots\pi(n{-}1)}\over T_{\pi(i+1)}} = \sum_{j=i+2}^{n-1} U_{\pi(j)}{T_{\pi(j)}\over T_{\pi(i+1)}}\leq \sum_{j=i+2}^{n-1} U_{\pi(j)}\leq 1 - \sum_{j=1}^{i+1}U_{\pi(j)}
\end{equation*}
From this follows  by $C_{\pi(i)}=U_{\pi(i)}T_{\pi(i)}$ and because of our task reordering we have $j \geq i, T_{\pi(j)}\leq T_{\pi(i)}$. Furthermore, we use $\sum_{j=1}^{n-1}U_{\pi(j)} < 1:$
\begin{multline}\label{eq:unique}
{C_{\pi(i{+}1)\dots\pi(n{-}1)}+C_{\pi(i{+}2)\dots\pi(n{-}1)}\over T_{\pi(i+1)}} = U_{\pi(i+1)}+2 \sum_{j=i+2}^{n-1} U_{\pi(j)}{T_{\pi(j)}\over T_{\pi(i+1)}}\\  \leq U_{\pi(i+1)}+2\sum_{j=i+2}^{n-1} U_{\pi(j)} <  2 -U_{\pi(i+1)} -2\sum_{j=1}^{i}U_{\pi(j)}
\end{multline}
From these formulas, situations can now be derived in which it can be guaranteed that the constraint system has an unique or no solution.
\begin{enumerate}
	\item For $T_{\pi(2)}>  T_{\pi(3)}>\dots.>  T_{\pi(n-1)}$ we have $T_{\pi(j)}/T_{\pi(i)} \leq 1/2$ for $j \geq i+2$. Hence 
	\begin{multline*}
	1\leq i \leq n-2,\ \ {C_{\pi(i{+}1)\dots\pi(n{-}1)}+C_{\pi(i{+}2)\dots\pi(n{-}1)}\over T_{\pi(i+1)}} =\\ U_{\pi(i+1)}+2 \sum_{j=i+2}^{n-1} U_{\pi(j)}{T_{\pi(j)}\over T_{\pi(i+1)}} < 1
	\end{multline*}
	\item By \eqref{eq:unique} we also get a unique solution if $1 < U_{\pi(i+1)}+ 2\sum_{j=1}^{i}U_{\pi(j)} $. This constraint is met if $\sum_{j=1}^{i}U_{\pi(j)} \geq 0.5$, i.e., if the tasks still to be processed contribute a total utilization $<0.5$. Therefore, if the \emph{worst-case response time} of a task $\tau_n$ is to be determined for which the  tasks $\tau_{\pi(1)\dots\pi(n-1)}$ have a utilization $< 0.5$, then the solution is unique. 
\end{enumerate}
We now assume that after determining a value $m_i$ two values for $m_{i+1}$ are possible. This corresponds to case 4 in \eqref{eq:gamma}. To keep the algorithm efficient, we select one of these values by determining the length of the interval $\left[m_{n.lb.i{+}1}, \right.\\ \left. m_{n.ub.i{+}1}\right]$ for the two values and then selecting the value with the larger interval length.

\subsubsection{Algorithm}
With the following algorithm we determine the values $\textbf{m}=\left[m_1, m_2, \dots., m_{n-1}\right]$ and  the maximum virtual jitter $ J'_{\pi(n-1)}$ according to \eqref{eq:vJ}.
 
%\begin{algorithm}
 
\textbf{Input}:\  A task system with the parameters $\textbf{T}=\left[T_{\pi(1)}, T_{\pi(2)}, \dots.,T_{\pi(n-1)}\right]$, $\textbf{C}=\left[C_{\pi(1)}, C_{\pi(2)}, \dots., C_{\pi(n-1)}\right]$, $\textbf{J}=\left[J_{\pi(1)}, J_{\pi(2)}, \dots., J_{\pi(n-1)}\right]$. The tasks are ordered by non-increasing periods i.e. $T_{\pi(1)} \geq T_{\pi(2)} \geq \ldots \geq T_{\pi(n-1)}$. Tasks with equal periods are arbitrarily ordered.

\textbf{Output}:\ 'infeasible' or $ J'_{max}$ and $\textbf{m}$ 

\textbf{Variables}:

 $m.T_{n-1.lb.i} \gets_{def} T_{\pi(n-1)}m_{n-1.lb.i}$; \Comment{newly introduced variables} 

  $m.T_{n-1.ub.i} \gets_{def} T_{\pi(n-1)}m_{n-1.ub.i}$;
	
 $q_{lb.i}, q_{ub.i};$  \Comment{auxiliary variables}
\begin{algorithmic}%[1]
\State $C_{\pi(n)\dots\pi(n-1)} \gets 0$

\For{$i=n-1\dots1$} 
\State $C_{\pi(i)\dots\pi(n-1)} \gets C_{\pi(i+1)\dots\pi(n-1)}+C_{\pi(i)}$ \EndFor
\State $m_1\gets 1$
\State $m.T_{n-1.lb.1} \gets T_{\pi(1)} +T_{\pi(n-1)}\left\lceil \frac{J_{\pi(1)}-J_{\pi(n-1)}}{T_{\pi(n-1)}}\right\rceil$ \Comment{Eq. \eqref{eq:task1}}
\State $m.T_{n-1.ub.1} \gets T_{\pi(1)} +T_{\pi(n-1)} \left\lfloor \frac{J_{\pi(1)}-J_{\pi(n-1)}+C_{\pi(2)\dots\pi(n-1)}}{T_{\pi(n-1)}}\right\rfloor$ \Comment{Eq. \eqref{eq:task1}}
\If {$m.T_{n-1.lb.1}>m.T_{n-1.ub.1}$}
\Return 'infeasible'  
\EndIf
\For{$i=2 \dots. n{-}2$}
\State $m_{lb.i} \gets \left\lceil \frac{m.T_{n-1.lb.i-1}+J_{\pi(n-1)}-J_{\pi(i)}-C_{\pi(i+1)\dots\pi(n-1)}}{T_{\pi(i)}}\right\rceil;$ \Comment{Eq. \eqref{eq:lower_bound}}

\State $m_{ub.i} \gets \left\lfloor \frac{m.T_{n-1.ub.i-1}+J_{\pi(n-1)}-J_{\pi(i)}}{T_{\pi(i)}}\right\rfloor;$ \Comment{Eq. \eqref{eq:m.i+1.3}}

\State $q_{lb.i} \gets T_{\pi(n-1)}\left\lceil \frac{J_{\pi(i)}-J_{\pi(n-1)}}{T_{\pi(n-1)}}\right\rceil;$
\State $q_{ub.i} \gets T_{\pi(n-1)}\left\lfloor  \frac{J_{\pi(i)}-J_{\pi(n-1)}+C_{\pi(i+1)\dots\pi(n-1)}}{T_{\pi(n-1)}}\right\rfloor$
\If {$m_{lb.i}=m_{ub.i}$} 
\State $m_i \gets m_{lb.i}$

 \State $m.T_{n-1.lb.i} \gets \max\left(T_{\pi(i)} m_i+q_{lb.i},m.T_{n.lb.i-1} \right)$ \Comment{Eq. \eqref{eq:lb.i+1}}

 \State $m.T_{n-1.ub.i} \gets \min\left(T_{\pi(i)} m_i+q_{ub.i},m.T_{n.ub.i-1} \right)$  \Comment{Eq. \eqref{eq:ub.i+1}}
\If {$m.T_{n-1.lb.i}>m.T_{n-1.ub.i}$} \Return 'infeasible' 
\EndIf																										
\ElsIf {$m_{lb.i}>m_{ub.i}$ }
\Return 'infeasible'
\Else \Comment{Evaluating the 2 feasible values $m_{lb.i}, m_{ub.i}$}
\State $m.T_{n-1.lb.0.i} \gets \max\left(T_{\pi(i)} m_{lb.i}+q_{lb.i},m.T_{n.lb.i-1} \right)$ \Comment{Eq. \eqref{eq:lb.i+1}} 
\State $m.T_{n-1.lb.1.i} \gets \max\left(T_{\pi(i)} m_{ub.i}+q_{lb.i},m.T_{n.lb.i-1} \right);$ \Comment{Eq. \eqref{eq:lb.i+1}} 

\State $m.T_{n-1.ub.0.i} \gets \min\left(T_{\pi(i)} m_{lb.i}+q_{ub.i},m.T_{n.ub.i-1} \right)$ \Comment{Eq. \eqref{eq:ub.i+1}}

\State $m.T_{n-1.ub.1.i} \gets \min\left(T_{\pi(i)} m_{ub.i}+q_{ub.i},m.T_{n.ub.i-1} \right);$ \Comment{Eq. \eqref{eq:ub.i+1}}

 \State $\textit{diff}_0 \gets m.T_{n.ub.0.i}-m.T_{n.lb.0.i} $
\State $\textit{diff}_1 \gets m.T_{n.ub.1.i}-m.T_{n.lb.1.i}$
\If {$\textit{diff}_0 > \textit{diff}_1$}  \Comment{Comparing the remaining interval lengths}
\State $m_i \gets m_{lb.i}$

\State $m.T_{n-1.lb.i} \gets m.T_{n-1.lb.0.i}$
\State $m.T_{n-1.ub.i} \gets m.T_{n-1.ub.0.i}$ 

\Else 
\State $m_i \gets m_{ub.i}$
\State $ m.T_{n-1.lb.i} \gets m.T_{n-1.lb.1.i}$ 
\State $m.T_{n-1.ub.i} \gets m.T_{n-1.ub.1.i} $

\EndIf
\If {$m.T_{n-1.lb.i} > m.T_{n-1.ub.i}$} \Return 'infeasible'
\EndIf 
\EndIf
\EndFor

\State $ J'_{max} \gets J_{\pi(n-1)}+ m_{n-1} T_{\pi(n-1)} $ \Comment{for using in \eqref{eq:interf}}

\State \Return $\left(\textbf{m}, J'_{max}\right)  $

\end{algorithmic} 
$$ $$
\textbf{Example}
$$\textbf{T}=(240,120,120,20,10);\textbf{C}=(1,50,50,1,1);\textbf{J}=(167,119,0,0,0);$$
 Since the tasks are in the right order we have $\pi(1)=1, \pi(2)=2, \pi(3)=3, \pi(4)=4, \pi(n-1)=5$

$m_1=1; (C_{2\dots5}, C_{3\dots5}, C_{4\dots5}, C_{5\dots5},C_{6\dots5})=(102,52,2,1,0); $
\begin{equation*}
m.T_{n.lb.1}:= T_1+T_5\left\lceil (J_1-J_5)/T_5\right\rceil=240+10\left\lceil 167/10\right\rceil=410
\end{equation*}
\begin{equation*}
m.T_{n.ub.1}= T_1+T_5\left\lfloor (J_1-J_5+C_{2\dots5})/T_5\right\rfloor=240+10\left\lfloor (167+102)/10\right\rceil=500
\end{equation*}
In the example we have at the beginning an interval $\left[41,50\right]$ for $m_5$ (using $T_5=10$).

$$i=2$$
\begin{equation*}
m_{lb.2}=\left\lceil (m.T_{n.lb.1}+J_5-J_2-C_{3\dots5})/T_2\right\rceil= \left\lceil (410-119-52)/120\right\rceil=2
\end{equation*}
\begin{equation*}
m_{ub.2}=\left\lfloor (m.T_{n.ub.1}+J_5-J_2)/T_2\right\rfloor= \left\lfloor (500-119)/120\right\rfloor=3
\end{equation*}

So there are two possible values for which the $m_5$ limits are now determined.  The auxiliary variables $q_{lb.2}$ and $q_{ub.2} $ have the values:
$q_{lb.2}:= T_{5}\left\lceil (J_{2}-J_{5})/T_{5}\right\rceil=10 \left\lceil 119/10\right\rceil =120;$
$q_{ub.2}:=T_{5}\left\lfloor  (J_{2}-J_{5}+C_{3\dots5})/T_{5}\right\rfloor=10 \left\lfloor  (119+52)/10\right\rfloor=170; $

We determine the limits of ${m_5}$ for $m_{lb.2}=2$:
\begin{equation*}
m.T_{n.lb.0.2}=\max(T_2 m_{lb.2}+q_{lb.2}, m.T_{n.lb.1})= \max(120*2+120, 410)= 410
\end{equation*}
\begin{equation*}
m.T_{n.ub.0.2}=\min(T_2 m_{lb.2}{+}q_{ub.2}, m.T_{n.ub.1}){=} \min(120*2+170, 500) =410
\end{equation*}
We get $m_5=m.T_{n.ub.0.2}/T_5=41$ and $\textit{diff}_1 =0;$

Now we determine the limits of $m_5$  for  $m_{ub.2}=3.$
\begin{equation*}
m.T_{n.lb.1.2}=\max(T_2 m_{ub.2}{+}q_{lb.2}, m.T_{n.lb.1})= \max(120*3{+}120, 410)=480
\end{equation*}
\begin{equation*}
m.T_{n.ub.1.2}=\min(T_2 m_{ub.2}+q_{ub.2}, m.T_{n.ub.1})=(360+170, 500)=500
\end{equation*}
In this case we have $m_5 \in [48,50]$, i.e. $\textit{diff}_2=500-480=20$.
We therefore select $m_2=3$ and set $m.T_{n.lb.2}=m.T_{n.lb.1.2}=480;m.T_{n.ub.2}=m.T_{n.ub.1.2}=500;$

$$i:=3$$
\begin{equation*}
m_{lb.3}=\left\lceil (m.T_{n.lb.2}+J_5-J_3-C_{4\dots5})/T_3)\right\rceil=\left\lceil(480-2)/120 \right\rceil =4
\end{equation*}
\begin{equation*}
m_{ub.3}=\left\lfloor (m.T_{n.ub.2}+J_5-J_3)/T_3)\right\rfloor=\left\lfloor 500/120 \right\rfloor =4
\end{equation*}
The lower and  upper limits are equal, i.e. $m_3=4$. For the next iteration we determine
\begin{equation*}
m.T_{n.lb.3}= \max(T_3 m_3 + T_5 \lceil (J_3-J_5)/T_5 \rceil, m.T_{n.lb.2}) = 480
\end{equation*}
\begin{equation*}
m.T_{n.ub.3}= \min(T_3 m_3+T_5\left\lfloor (J_3-J_5+C_{4\dots5})/T_5\right\rfloor,m.T_{n.ub.2}) = 480
\end{equation*}
This means that we have only one value left for $m_5$ namely 48.
$$i=4$$
\begin{equation*}
m_{lb.4}=\left\lceil (m.T_{n.lb.3}+J_5-J_4-C_{5\dots 5})/T_4)\right\rceil=\left\lceil(480-1)/20 \right\rceil =24
\end{equation*}
\begin{equation*}
m_{ub.4}=\left\lfloor (m.T_{n.ub.3}+J_5-J_4)/T_5)\right\rfloor=\left\lfloor 480/20 \right\rfloor =24
\end{equation*}
Since $m_{lb.4}=m_{ub.4}$ we get $m_4 = 24$.
\begin{equation*}
m.T_{n.lb.4}= \max(T_4 m_4+T_5\left\lceil (J_4-J_5)/T_5\right\rceil, m.T_{n.lb.3}) = 480
\end{equation*}
\begin{equation*}
m.T_{n.ub.4}= \min(T_4 m_4+T_5\left\lfloor (J_4-J_5)/T_4\right\rfloor, m.T_{n.ub.3}) = 480
\end{equation*}
The remaining value from iteration 3 is still valid: $m_5=48$
$$i=5$$
$m_5 =480/10 = 48$ 

The final result is:
$J'_{max}:=480;$
$\textbf{m}:=(1,3,4,24,48);$

\section{Experiments for task systems with jitter}
Our algorithm is not suitable for arbitrary jitter values, because the necessary restrictions are too strict. For example, if we create task sets pseudo-randomly and allow all jitter $J_{\pi(i)}$ within the intervals $[0, \alpha T_{\pi(i)}]$ with $0<\alpha\leq 1$, we will only get allowed jitter values for a very small percentage of real-time systems (<2\%). This is true even if we only consider a few tasks (e.g. 5) and high total utilization (e.g. 0.95) as shown in \ref{fig:4}.  The usefulness of our algorithm must therefore be proven by practical examples, for which we refer to future work. 

\begin{figure*}[!htp]
\centering
\includegraphics[width=0.70\textwidth]{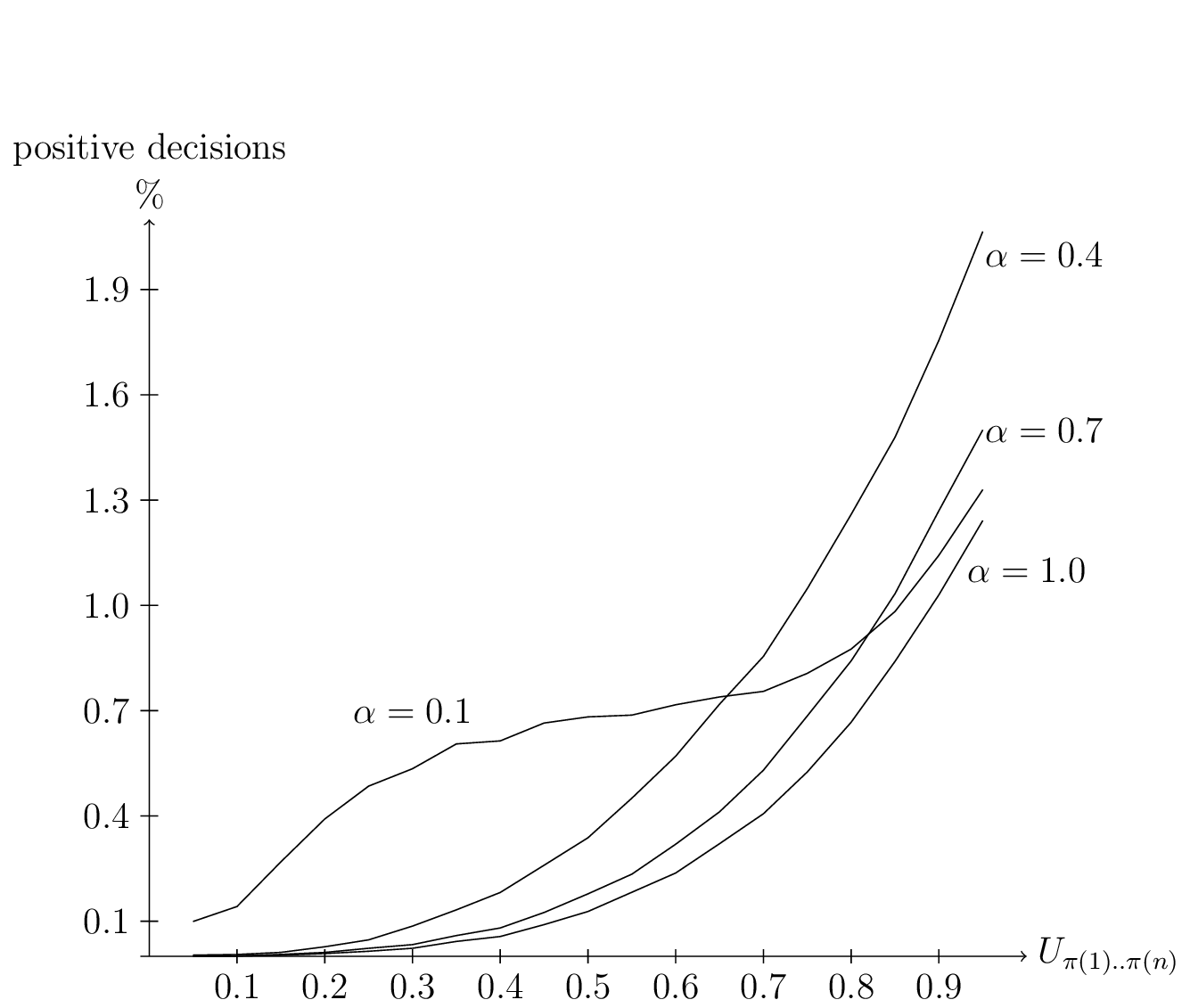}
\caption{Percentage number of feasible task sets  with jitters $J_{\pi(i)} \in [0,\alpha T_{\pi(i-1)}]$ based on 1000000 pseudo-randomly produced task sets for each $\alpha$}	
\label{fig:4}	
\end{figure*}

Of greater interest is an answer to the question of how good the quality of the heuristic component is in our algorithm. In our experiments we use the following rule to compute the periods. The first period $T_1=10$ is chosen arbitrarily. This has no effect on the meaningfulness of the simulation as all response time bounds should be independent of a scaling factor applied to all parameters. The periods of the other tasks are produced iteratively by pseudo-randomly selecting a factor from [1..4]. In order to get the utilization values we use the algorithm UUniFast, as described in \cite{Butta2005}.
The  values $C_{\pi(i)}=T_{\pi(i)} U_{\pi(i)}$ are not of type integer which is not relevant in this case.

The jitter values are produced observing  the constraints  \eqref{eq:viJup},...,\eqref{eq:ineq}. First, we determine $J_{\pi(1)}$ by pseudo-randomly selecting a value from $[0,T_{\pi(1)}-1]$. From this we get $J'_{\pi(1)}=T_{\pi(1)}+ J_{\pi(1)}$.

We combine \eqref{eq:viJlo} and \eqref{eq:viJup} for $i=1$ and determine $J'_{\pi(n-1)}$ by selecting pseudo-randomly a value from $[J'_{\pi(1)}, J'_{\pi(1)}+C_{\pi(2)..\pi(n-1)}]$. Then we get $J_{\pi(n-1)}=J'_{\pi(n-1)} \mod T_{\pi(n-1)} $.

Since we now know the value $J'_{\pi(n-1)}$ we can determine by \eqref{eq:viJlo} and \eqref{eq:viJup}  the other values $J'_{\pi(i)}$ selecting pseudo-randomly a value from $[J'_{\pi(n-1)}-C_{\pi(i+1)..\pi(n-1)}, J'_{\pi(n-1)}]$. It follows $J_{\pi(i)} = J'_{\pi(i)} \mod T_{\pi(i)} $.

Such a task set fulfills the constraints \eqref{eq:viJlo},..., \eqref{eq:ineq} and our algorithm should be able to characterize it as feasible and should determine $m_i$ values and $J_{max}$. If it is not successful, this is due to the heuristic part of the algorithm which selects in these cases the wrong value $m[i]$ .

In our experiment we let the total utilization grow in steps of 0.5 and created 2000000 task sets with $n-1=14$ tasks for each of these values. We found that up to a total utilization of $U_{\pi(1),..\pi(n-1)}$$  = 0.75$ all task set are correctly classified. For larger values of the utilization we have few task sets that are incorrectly classified as infeasible. Note that a larger total utilization means  larger execution times and therefore larger ranges for the jitter values. Table \ref{tab3} shows the concrete number of incorrectly classified task sets.
\begin{table}[htbp]
\caption{Effectivity of the Jitter-Heuristic }
\begin{center}
\begin{tabular}{|c|c|c|c|c|c|c|}

\hline

\rule[-1mm]{0mm}{4.5mm}\textbf{$U_{\pi(1)..\pi(13)}$} & 0.05-0.75 &0.8& 0.85&0.9&0.95  \\
\hline
\rule[-0.5mm]{0mm}{4.1mm}falsely classified tasksets& 0 &6 & 10 &17&33 \\
\hline

\end{tabular} 
\label{tab3}
\end{center}
\end{table}

%\bibliographystyle{plain} 
%\bibliography{RealTimePaper} 

%\end{document}

\section{Conclusions}
Because of the manifold practical applications of task systems with harmonic tasks it is important to take advantage of the special features resulting from the divisibility of periods by all smaller periods. For example, response time analysis is possible in polynomial time, while in the general case it has pseudo-polynomial complexity. We have introduced a new algorithm that calculates the exact \emph{worst-case response time} of a task in linear time when the higher-priority tasks are ordered by non-increasing periods. Our algorithm has another advantage, which is that the task model can be extended to practical requirements. We have made this more concrete using the example of release jitters, which previous special algorithms for harmonic tasks could not handle. However, we cannot process all jitter-aware task systems with harmonic periods with it and  we have therefore proposed a linear algorithm to check the jitter values for feasibility.

\bibliography{RealTimePaper} 
\bibliographystyle{plain} 
\end{document}